\documentclass[11pt]{article}

\usepackage{amsfonts,amssymb, latexsym, amsmath, amsthm,mathrsfs,
            verbatim,geometry}
\usepackage{graphics}
\usepackage{slashed}
\usepackage{url,color}
\usepackage{enumerate}

\usepackage{hyperref}
\hypersetup{
    colorlinks,
    citecolor=red,
    filecolor=green,
    linkcolor=blue,
    urlcolor=black
}

\usepackage{caption}
\usepackage{tikz}
\usetikzlibrary{calc}
\usetikzlibrary{arrows.meta}

\begin{document}

\arraycolsep = 0.3\arraycolsep
\def\R{\mathbb R}
\def\C{\mathbb C}
\def\N{\mathbb N}
\def\Z{\mathbb Z}
\newcommand\J{\mathscr J}
\def\A{\mathscr{A}}
\def\p{\left\langle v\right\rangle}
\def\tp{\left\langle \tilde v\right\rangle}
\def\u{\left\langle u\right\rangle}
\def\T{\mathbb{T}}
\def\energy{\mathcal{E}}
\def\D{\mathcal{D}}
\def\be{\begin{equation}}
\def\ee{\end{equation}}
\def\bea{\begin{eqnarray}}
\def\eea{\end{eqnarray}}
\def\beas{\begin{eqnarray*}}
\def\eeas{\end{eqnarray*}}
\def\supp{\mathrm{supp}\,}
\def\g{\partial}
\def\t{\bar{\partial}}
\def\a{{\bf a}}
\def\K{\mathbb{S}}
\def\l{\lambda}
\def\pa{\partial }
\def\energynergy{{\overline{\mathcal E}^\varphi}}
\def\w{w_{\delta}}
\def\S{\mathcal S}
\def\der{\pa_r^a\slashed\nabla^\beta}
\def\car{\pa^\nu}
\def\d{\text{div}}
\def\sn{\slashed\nabla}
\def\M{\mathscr M}
\def\N{\mathscr N}
\def\drho{\delta\rho}
\def\dl{\delta\lambda}
\def\bcr{\begin{color}{red}}
\def\ec{\end{color}}
\def\lv{\left\vert}
\def\rv{\right\vert}

\def\norm{\mathcal S^N}
\def\vortnorm{\mathcal B^N}
\def\energy{\mathcal E^N}
\def\dissipation{\mathcal D^N}
\def\pa{\partial}
\newcommand{\prfe}{\hspace*{\fill} $\Box$

\smallskip \noindent}

\sloppy
\newtheorem{maintheorem}{Theorem}
\newtheorem{theorem}{Theorem}[section]
\newtheorem{definition}[theorem]{Definition}
\newtheorem{proposition}[theorem]{Proposition}
\newtheorem{example}[theorem]{Example}
\newtheorem{corollary}[theorem]{Corollary}
\newtheorem{lemma}[theorem]{Lemma}
\newtheorem{remark}[theorem]{Remark}

\renewcommand{\theequation}{\arabic{section}.\arabic{equation}}

\title{Turning point principle for relativistic stars}
\author{Mahir Had\v zi\'c\thanks{Department of Mathematics,
        University College London, London, UK} , 
        Zhiwu Lin\thanks{Department of Mathematics,
        Georgia Institute of Technology, Atlanta, USA}}
        
\maketitle

\begin{abstract}
Upon specifying an equation of state, spherically symmetric steady states of the Einstein-Euler system are embedded in 1-parameter families of solutions, characterized by the
value of their central redshift. In the 1960's Zel'dovich~\cite{Ze1963} and Wheeler~\cite{HaThWaWh} formulated a turning point principle which states that the spectral stability can be exchanged to instability and vice versa only at the extrema of mass along the mass-radius curve.  Moreover the bending orientation at the extrema determines whether a growing mode is gained or lost. We prove the turning point principle and provide a detailed description of the linearized dynamics. One of the corollaries of our result is that the number of growing modes grows to infinity as the central redshift increases to infinity.
\end{abstract}

\tableofcontents

\section{Introduction}

In this work we rigorously establish the turning point principle for  radial relativistic stars along the so-called mass-radius curve of 1-parameter family of stationary solutions, see Theorem~\ref{T:TPP}. This principle was formulated by Zel'dovich~\cite{Ze1963} and Wheeler, see~\cite{HaThWaWh} (pages 60--66), and it is also referred to as the {\em M(R)-method}. In the radial setting this is a powerful tool predicting the exact number of unstable eigenmodes for the linearized radial Einstein-Euler system around its dynamic equilibria, based solely on the the location of the equilibrium on the mass-radius curve. 

In our previous work~\cite{HaLiRe} jointly with Rein, among other things we introduced the so-called separable Hamiltonian formulation of the linearized Einstein-Euler system, which highlights the symplectic structure in the problem. This proved crucial to a refined understanding of the linearized flow and its decomposition into invariant subspaces, where we used a general framework developed recently by Lin and Zeng~\cite{LinZeng2020}. The second main result of this paper is a precise index formula which expresses the number of unstable modes as the difference of the negative Morse index of a certain Schr\"odinger type operator~\eqref{E:SIGMAKAPPADEF} and a quantity we call the {\em winding index} which reflects the winding properties of the mass-radius curve, see Definition~\ref{D:INDEX} and Theorem~\ref{T:NEGATIVEMODES}. This result completes a related result from~\cite{HaLiRe} by including equilibria with certain exceptional values of the central redshift parameter.

The unknowns are the 4-dimensional spacetime $M$ and a Lorentzian metric $g$ with signature $(-,+,+,+)$, while the fluid unknowns are the {\em density} $\rho$, {\em pressure} $p$, and the {\em 4-velocity} $u^\mu$, $\mu=0,1,2,3$, which is normalized to be a future pointing unit timelike vector 
\be\label{E:VELOCITYNORM}
g^{\mu\nu}u_\mu u_\nu = -1.
\ee

The unknwons are dynamically coupled  through the Einstein field equations
\begin{align}
G_{\mu\nu} = 8\pi T_{\mu\nu}, \ \ \mu,\nu=0,1,2,3,
\end{align}
where $G_{\mu\nu}$ is the Einstein tensor and $T_{\mu\nu}$ the energy-momentum tensor given by 
\begin{align}
T_{\mu\nu} = (\rho+p)g_{\mu\nu} + p u_{\mu}u_\nu, \ \ \mu,\nu=0,1,2,3.
\end{align}

We shall work in radial symmetry and use the Schwarzschild coordinates where the metric takes the form
\be \label{E:METRIC}
ds^2=-e^{2\mu(t,r)}dt^2+e^{2\lambda(t,r)}dr^2+
r^2(d\theta^2+\sin^2\theta\,d\varphi^2).
\ee
and the $4$-velocity takes the form $u=(u^0,u,0,0)$. By~\eqref{E:VELOCITYNORM} and~\eqref{E:METRIC} we get 
\begin{align}
u^0= e^{-\mu}\sqrt{1+e^{2\lambda} u^2} =: e^{-\mu} \u.
\end{align}
The field equations become~\cite{HaLiRe}
\be
e^{-2\lambda} (2 r \lambda' -1) +1 =
8\pi r^2 \left(\rho + e^{2\lambda}(\rho + p) u^2\right),\label{E:LAMBDAEQN}
\ee
\be
e^{-2\lambda} (2 r \mu' +1) -1 =
8\pi r^2 \left(p + e^{2\lambda}(\rho + p) u^2\right),\label{E:MUEQN}
\ee
\be
\dot \lambda =
- 4 \pi  r e^{\mu + 2\lambda}  \u \, u\, (\rho + p)
\label{E:LAMBDADOTEQN},
\ee
\be
e^{- 2 \lambda} \left(\mu'' + (\mu' - \lambda')(\mu' + \frac{1}{r})\right)
- e^{-2\mu}\left(\ddot \lambda + \dot\lambda (\dot \lambda - \dot \mu)\right)
= 8 \pi p. \label{ee2ndo} 
\ee
The Euler equations become
\bea 
\dot \rho + e^{\mu} \frac{u}{\u} \rho' 
+ \left(\rho + p\right)
\biggl[\dot \lambda + e^\mu \frac{u}{\u}
\left(\lambda' + \mu'+ \frac{2}{r}\right)
+ e^\mu \frac{u'}{\u}
+ e^{2\lambda} \frac{u}{\u}
\frac{\dot \lambda u + \dot u}{\u}\biggr]
&=&
0,\qquad \label{rhod}\\
(\rho +p) \left[e^{2\lambda} \left( \dot u + 2 \dot \lambda u\right)
  + e^{\mu} \u \mu'
+ e^{\mu+2\lambda} \frac{u}{\u}
\left(u' + \lambda' u\right)\right]
+ e^{\mu}\u p' + e^{2\lambda} u \,\dot p
&=&
0.\qquad \label{ud}
\eea

There are however too many fluid unknowns, so to close 
the system we impose a {\em barotropic} equation of state between the pressure and the density. We assume that 
\[
p = P(\rho)
\]
where the state function $P$ satisfies the following assumptions:
\begin{enumerate}
\item[(P1)]
$P \in C^1([0,\infty[)$ with $P'(\rho) >0$ for $\rho>0$,
$P(0)=0$, 
\item[(P2)]
For some $\zeta>0$ there exists a $C^1$-function $f:[0,\zeta]\to\mathbb R$ such that $f(0)=0$ and
\begin{align}
P(\rho) = k \rho^\gamma (1+ f(\rho))
\end{align}
for some $\frac43<\gamma<2$.
This in particular implies that there exists a constant $c_1$ such that $|f(\rho)|\le c_1\rho$ on $[0,\zeta]$ and therefore
\begin{align}\label{E:P2}
P(\rho) = k \rho^\gamma + O_{\rho\to0+}(\rho^{\gamma+1}).
\end{align} 

\item[(P3)] 
There exist the inverse of $P$ on $[0,\infty)$ and  constants $0<c_s^2\le 1$, $c_2>0$ such that   
\be \label{Pass2}
|p - c_s^2\rho| \leq c_2 p^{1/2}\ \mbox{for all}\ p>0.
\ee
\item[(P4)]
For any $\rho>0$ we have 
\[
0<\frac{dP}{d\rho}\le 1.
\]
This is a causality assumption and states that the speed of sound inside the star never exceeds the speed of light.
\end{enumerate}

Assumptions (P1)--(P4), or some qualitatively similar version of those, are quite commonly used in the description of gaseous stars in relativistic astrophysics, see~\cite{He,Ma1998, HaLiRe} and references therein. For a detailed study of the equations of states for neutron stars see~\cite{HaPoYa2007}. Assumption (P2) states that in the region close to vacuum ($0<\rho\ll1$) the equation of state is effectively described by the classical polytropic power law $P(\rho)=k\rho^\gamma$. On the other hand, in the regime where the density is very large ($\rho\gg1$) assumption (P3) states that to the leading order $P(\rho)=c_s^2\rho$. Here $0<c_s\le1$, which is also a consequence of the causality assumption stated in assumption (P4). We also observe that assumptions (P1) and (P4) imply that $P(\rho)\le \rho$.

We shall refer to the system of equations~\eqref{E:LAMBDAEQN}--\eqref{ud} together with assumptions (P1)--(P4) as the spherically symmetric Euler-Einstein system and use the abbreviation EE-system.

There are two basic conserved quantities - the ADM mass
\begin{align}
M(\rho) = 4\pi \int_0^\infty \rho(r) \,r^2dr
\end{align}
and the total particle (baryon) number 
\begin{align}
N(\rho) = 4\pi  \int_0^\infty e^\l r^2 n(\rho)\,r^2dr, \ \ n(\rho) := \exp\left(\int_1^\rho \frac{ds}{s+P(s)}\right).
\end{align}

We look for compactly supported steady states of the EE-system~\eqref{E:LAMBDAEQN}--\eqref{ud} satisfying $u=0$. Equation~\eqref{rhod} is then automatically satisfied and equation~\eqref{ud} reduces to the famous Tolman-Oppenheimer-Volkov relation:
\be\label{E:TOV}
(\rho+p)\mu' + p' = 0.
\ee
We define 
\be\label{E:QDEF}
Q(\rho) := \int_0^\rho \frac{P'(s)}{s+P(s)} ds,\ \rho\geq 0, 
\ee
so that~\eqref{E:TOV} immediately implies
\[
Q(\rho(r))+\mu(r) = \text{const.}
\]
We introduce the unknown $y(r)=\text{const.}-\mu(r)$, so that $\rho$ can now be expressed through
\be \label{E:GDEF}
\rho = g(y) := \left\{
\begin{array}{ccl}
  Q^{-1}(y)&,&y>0,\\
  0&,&y\leq 0.
\end{array}
\right.
\ee
The field equation~\eqref{E:LAMBDAEQN} with $u=0$ can be rewritten in the form $\pa_r\left(r-e^{-2\l}r\right)=8\pi r^2\rho$, which immediately yields
\begin{align}\label{E:LAMBDARELATION}
e^{-2\l(r)} = 1 - \frac{2m(r)}{r}, \ \ m(r) = \int_0^r 4\pi s^2\rho(s)\,ds.
\end{align}
Plugging the above into the field equation~\eqref{E:MUEQN} with $u=0$, we finally obtain 
the fundamental steady state equation satisfied by $y$:
\be \label{E:YEQN}
y'(r)= - \frac{1}{1-2 m(r)/r} \left(\frac{m(r)}{r^2} + 4 \pi r p(r)\right).
\ee
Here $p$ is given in terms of $y$ by 
the relations
\begin{align}
p(r) & = h(y(r))= P(g(y(r))).
\end{align}


The existence of compactly supported steady states follows for example from the work of Ramming and Rein~\cite{RaRe}, which we state in the following 
proposition for readers' convenience.
\begin{proposition}[\cite{RaRe}]\label{P:EXISTENCE}
Under the assumptions (P1)--(P4) on the equation of state
for any central value
\be\label{E:YINITIAL}
y(0)=\kappa>0
\ee
there exists a unique smooth solution $y=y_\kappa$
to \eqref{E:YEQN}, which is defined on $[0,\infty)$
and has a unique zero at some radius $R_\kappa>0$. The value $R_\kappa$ is the radius of the star.
\end{proposition}


\begin{remark}
The existence of compactly supported radial steady star solutions to the Einstein-Euler system is well-known, see~\cite{He,NiUg,Si,RaRe} and references therein.  The assumptions on the equation of state, in particular the lower bound on $\gamma$ in {\em (P1)} can be relaxed, and the finite extent property can also be shown in different ways~\cite{NiUg,He,RaRe}.
\end{remark}

Given $y_\kappa$, we define $\rho_\kappa$ and $\l_\kappa$ via~\eqref{E:GDEF} and~\eqref{E:LAMBDARELATION} respectively. The metric coefficient $\mu_\kappa$ is then 
obtained through the formula
\begin{align}
\mu_\kappa(r) = \mu_\kappa(R_\kappa)- y_\kappa(r), \ \ \mu_\kappa(R_\kappa) = \lim_{r\to\infty} y_\kappa(r).
\end{align}
For any $\kappa>0$ we refer to the triple $(\rho_\kappa,\mu_\kappa,\l_\kappa)$ as the steady state of the Euler-Einstein system.

\begin{remark}[Central redshift]
The central redshift $z$ of the star $(\rho_\kappa,\mu_\kappa,\l_\kappa)$ measures the redshift of a photon emitted at the center
of the star and received at its boundary. It is given by the formula
\begin{align}
z = \frac{e^{\mu_\kappa(R_\kappa)}}{e^{\mu_\kappa(0)}} - 1 = \frac{e^{y_\kappa(0)}}{e^{y_\kappa(R_\kappa)}} - 1 = e^\kappa -1.
\end{align}
Therefore $\kappa$ and $z$ are in a 1-1 relationship and, by slight abuse of terminology, we continue to call $\kappa$ the {\em central redshift} parameter.
\end{remark}

At the heart of our analysis is the formulation of the linearized flow as a separable Hamiltonian system derived in~\cite{HaLiRe}.
The natural function spaces contain weights that for each $\kappa>0$ depend on the solution $(\rho_\kappa,\mu_\kappa,\l_\kappa)$. 
An important role is played by the quantity
\begin{align}
\Psi_\kappa :=  e^{-\mu_\kappa} \frac{P'(\rho_\kappa)}{\rho_\kappa+p_\kappa}, \ \ r\in[0,R_\kappa).
\end{align}
It is easy to check using property (P1) that the function 
\begin{align}\label{E:PSIKAPPAINVERSE}
\Psi_\kappa^{-1} : = \begin{cases} e^{\mu_\kappa} \frac{\rho_\kappa+p_\kappa}{P'(\rho_\kappa)}, & r\in[0,R_\kappa] \\
0, & r> R_\kappa
\end{cases}
\end{align}
is $C^0$ on $[0,\infty)$. It is in fact slightly better - a simple consequence of the Hopf lemma is that close to the star boundary $\rho \sim (R_\kappa-r)^{\frac1{\gamma-1}}$ and therefore
\be\label{E:PSIKAPPAINVASYMP}
\Psi_\kappa^{-1}  \sim_{r\to R_\kappa} (R_\kappa-r)^{\frac{2-\gamma}{\gamma-1}},
\ee
where we have used~\eqref{E:PSIKAPPAINVERSE} and the property (P2).
Note that we abused the notation slightly by denoting $\Psi_\kappa^{-1}$ the extension of the reciprocal of $\Psi_\kappa$ on $[0,R_\kappa)$ to $[0,\infty)$.

\begin{definition}[Function spaces]\label{D:XDEF_EE}
Let the equation of state $\rho\to P(\rho)$ satisfy assumptions (P1)--(P4) and let $(\rho_\kappa,\mu_\kappa,\l_\kappa)$ be the 1-parameter family of steady states given 
by Proposition~\ref{P:EXISTENCE}.
\begin{itemize}
\item[(a)]
  The Hilbert space $X_\kappa$ is the space of all spherically symmetric
  functions in the weighted $L^2$ space on the set
  $B_\kappa=B_{R_\kappa}$ (the ball with radius $R_\kappa$ which is
  the support of $\rho _{\kappa}$) with weight
  $e^{2\mu_\kappa +\lambda_\kappa}\Psi_{\kappa }$ and the corresponding inner product,
  $Y_\kappa$ is the space of
  radial functions in $L^2 \left( B_\kappa\right)$, and the phase space for the
  linearized Einstein-Euler system is
  $ X_{\kappa}\times Y_\kappa.$  
\item[(b)]
  For $\rho\in X_\kappa$ the {\em induced modified potential $\bar \mu$}
  is defined as
  \begin{align}\label{E:MUBARDEF}
    \bar\mu(r) = \bar\mu_\rho(r)
    : =- e^{-\mu_\kappa-\l_\kappa}\int_r^\infty \frac1s\,
    e^{\mu_\kappa(s)+\l_\kappa(s)}(2s\mu_\kappa'(s)+1)\,\l(s)\,ds,
  \end{align}
  where
  \be\label{E:LAMBDAOFRHODEF}
  \l(r) = 4\pi \frac{e^{2\l_\kappa}}r \int_0^r s^2 \rho(s)\,ds, \ \ r\ge0
  \ee
  where $\rho$ is extended by $0$ to the region $r> R_\kappa$.
\item[(c)]
  The operators
  $L_\kappa \colon X_\kappa
  \to X_\kappa^\ast$ and $\mathbb L_\kappa \colon X_\kappa \to X_\kappa$
  are defined by
  \be\label{E:TILDELKAPPADEF}
  L_\kappa \rho : =
  e^{2\mu_\kappa+\lambda_\kappa}\Psi_{\kappa }\rho +e^{\mu_\kappa +\lambda_\kappa}\bar{\mu}_\rho.
  \ee
  \be\label{E:MATHCALLKAPPADEF}
  \mathbb L_\kappa \rho : = e^{-2\mu_\kappa-\lambda_\kappa}\Psi^{-1}_{\kappa}
L_\kappa \rho
  = \rho + e^{-\mu_\kappa}\Psi_\kappa^{-1}\bar{\mu}_\rho
  \ee
\end{itemize}
\end{definition}
Here the dual pairing is realized through the $L^2$-inner product, so that 
\[
\langle L_\kappa \rho, \bar\rho\rangle
= \left(\mathbb L_\kappa \rho, \bar\rho\right)_{X_\kappa},
\quad \rho,\bar\rho\in X_\kappa.
\]

As proved in~\cite{HaLiRe} (Section 5.2)
the formal linearization of the spherically symmetric Einstein-Euler system
  takes the separable Hamiltonian form 
  \begin{equation}\label{E:FIRSTORDERFLOW}
    \frac{d}{dt}
    \begin{pmatrix} \rho \\ v \end{pmatrix}
    =J^{\kappa }\mathcal L^{\kappa }
    \begin{pmatrix} \rho \\ v \end{pmatrix},    
  \end{equation}
where $\left( \rho ,v\right) \in X$, and 
\begin{equation}
  J^{\kappa } :=
  \begin{pmatrix}
  0 & A_{\kappa } \\ 
  -A_{\kappa }^{\prime } & 0
  \end{pmatrix},\quad 
  \mathcal L^{\kappa } :=
  \begin{pmatrix}
  L_{\kappa } & 0 \\ 
  0 & \mathrm{id}
  \end{pmatrix}.
\label{defn-J-K-kappa}
\end{equation}
Moreover $J^{\kappa}:X_\kappa^{\ast}\times Y_\kappa^\ast\rightarrow X_\kappa\times Y_\kappa$ and $\mathcal L^{\kappa}:X_\kappa\times Y_\kappa\rightarrow X_\kappa^{\ast}\times Y_\kappa^\ast$
are anti-self-dual and self-dual respectively. 
Here the operators $A_\kappa : Y_\kappa^\ast\to X_\kappa$ and its dual $A_\kappa':X_\kappa^\ast\to Y_\kappa$  
are given by 
\be \label{E:AKAPPADEF}
  A_{\kappa }v
  :=-\frac{1}{r^{2}} \frac{d}{dr}\left( r^{2}e^{-\frac{3}{2}\lambda_{\kappa }}
  n_{\kappa }^{\frac{1}{2}}v\right), \ \ \ \ 
 A_{\kappa }^\prime \rho :=e^{-\frac{3}{2}\lambda _{\kappa }}
 n_{\kappa }^{\frac{1}{2}} \frac d{dr}\rho .
\ee
Operators $A_\kappa$ and $A_\kappa'$ are densely defined and closed, see Section 5.2 of~\cite{HaLiRe}.  The conserved energy associated with~\eqref{E:FIRSTORDERFLOW} is given by 
\begin{align}
\mathcal E = 
\left(\mathcal L_\kappa \begin{pmatrix} \rho \\ v \end{pmatrix}, \begin{pmatrix} \rho \\ v \end{pmatrix} \right)_{X_\kappa\times Y_\kappa} 
= \left(L_\kappa\rho,\rho\right)_{X_\kappa} + \|v\|_{Y_\kappa}^2.
\end{align}

It is important to note that the first order formulation~\eqref{E:FIRSTORDERFLOW} can be equivalently replaced
by a second order formulation
\begin{align}\label{E:SECONDORDER}
\frac{d^2}{dt^2}v
  +A_{\kappa }^{\prime}L_{\kappa }A_{\kappa}v=0,
\end{align}
which at a formal level follows trivially from~\eqref{E:FIRSTORDERFLOW} by taking a time derivative. It is clear from~\eqref{E:SECONDORDER} that the steady state is spectrally stable if and only if the quadratic form 
\[
\langle A_{\kappa }^{\prime}L_{\kappa }A_{\kappa}v, v\rangle = \langle L_{\kappa }A_{\kappa}v, A_\kappa v\rangle
\]
is positive definite. In fact, the number of negative eigenvalues corresponds to the negative Morse index associated with the operator $L_\kappa\big|_{\overline{R(A_\kappa)}}$.
The space $\overline{R(A_\kappa)}\subset X_\kappa$ corresponds to the set of all {\em dynamically accessible perturbations}. It is not hard to see (Section 5.3 of~\cite{HaLiRe}) that 
\be\label{E:DYNACC}
\overline{R(A_\kappa)} = \left\{\rho\in X_\kappa\ \big| \int_{B_\kappa} \rho \,dx = 0\right\}.
\ee 
\begin{remark}
A simple consequence of the above discussion is the {\em Chandrasekhar stability critetrion}~\cite{Chandrasekhar1964,HaLiRe}: steady state $(\rho_\kappa,\mu_\kappa,\l_\kappa)$ is spectrally stable if and only if $\langle L_\kappa\rho,\rho\rangle\ge0$ for all $\rho\in X_\kappa$ with mean 0.
\end{remark}

\begin{remark}[Negative Morse index]
For a linear operator $L:H\to H^\ast$, $H$ a Hilbert space, the negative Morse index $n^-(L)$ of $L$ is by definition
the maximal dimension of subspaces of
$H$ on which $\langle L\cdot,\cdot\rangle <0$.  
\end{remark}

A crucial tool in our proof of the turning point principle is the so-called {\em reduced operator} discovered in~\cite{HaLiRe}:
\be\label{E:SIGMAKAPPADEF}
\Sigma_\kappa = - \frac1{4\pi}\Delta_\kappa - e^{\l_\kappa}\Psi_\kappa^{-1},
\ee
where
\be
\Delta_\kappa : =
\frac{e^{\mu_\kappa+\l_\kappa}}{ r^2}
\frac{d}{dr}\left(\frac{e^{-\mu_\kappa-3\l_\kappa}r^2}{2r\mu_\kappa'+1}
\frac{d}{dr} \left(e^{\mu_\kappa+\l_\kappa}\cdot\right)\right).
\ee
The operator $\Sigma_\kappa : \dot H^1_r \to (\dot H^1_r)^\ast$ is selfdual.
From~\eqref{E:MUBARDEF} and~\eqref{E:LAMBDAOFRHODEF} it is clear that $\Delta_\kappa \bar\mu_\rho = e^{\mu_\kappa+\l_\kappa} \rho$ for any $\rho\in X_\kappa$
and therefore
\be\label{E:FUNDAMENTAL}
L_\kappa \rho = -e^{\mu_\kappa}\Psi_\kappa \, \Sigma_\kappa \bar\mu_\rho, \ \ r\in[0,R_\kappa].
\ee

One of the key properties of the reduced operator is described in the following lemma.


\begin{lemma}\label{L:NEGMORSEINDEX}
Let the equation of state $\rho\to P(\rho)$ satisfy assumptions (P1)--(P4) and let $(\rho_\kappa,\mu_\kappa,\l_\kappa)$ be the 1-parameter family of steady states given 
by Proposition~\ref{P:EXISTENCE}. Then for every $\kappa>0$ we have 
\begin{align}\label{E:NEGMORSEINDEX}
n^-(L_\kappa\big|_{X_\kappa})=n^-(\mathbb L_\kappa) = n^-(\Sigma_\kappa).
\end{align}
\end{lemma}


\begin{proof}
By Theorem 5.15 in~\cite{HaLiRe} for every $\mu\in \dot H^1_r$ we have the bound 
$
\langle\Sigma_\kappa \mu,\mu\rangle \ge \langle L_\kappa\rho_\mu,\rho_\mu\rangle,
$
where
$
\rho_\mu = - e^{-\mu_\kappa}\Psi_\kappa^{-1}\mu\in X_\kappa.
$
Conversely, for any $\rho\in X_\kappa$ we have
$
\langle L_\kappa\rho,\rho\rangle \ge \langle\Sigma_\kappa \bar \mu_\rho,\bar\mu_\rho\rangle
$, where $\bar\mu_\rho$ is given by the formula~\eqref{E:MUBARDEF}. It then follows that 
$n^{\le}(L_\kappa\big|_{X_\kappa})=n^{\le}(\Sigma_\kappa)$. On the other hand, it follows from~\eqref{E:FUNDAMENTAL} 
that $\dim\text{ker}(L_\kappa) = \dim\text{ker}(\Sigma_\kappa)$ and this yields~\eqref{E:NEGMORSEINDEX}.
\end{proof}

The Newtonian limit of the Einstein-Euler system is the well-known gravitational Euler-Poisson system:
\begin{align}
\dot\rho + \text{div}(\rho {\bf u}) & = 0 \label{E:CONTINUITYEP}\\
\rho\left(\dot{\bf u}+{\bf u}\cdot\nabla {\bf u}\right) + \nabla p & = - \rho \nabla\phi \\
\Delta\phi & = 4\pi \rho, \ \ \lim_{|x|\to\infty}\phi(t,x) = 0. \label{E:POISSON}
\end{align}
Here $\rho$ is the fluid density, ${\bf u}$ the Newtonian 3-velocity, and $\phi$ the gravitational potential satisfying the Poisson equation~\eqref{E:POISSON}. Upon specifying an equation of state $p = P(\rho)$ one finds 1-parameter family of radial equilibria. The most famous among them are the compact Lane-Emden stars, associated with the so-called polytropic equation of state 
\be\label{E:POLYTROPIC}
P(\rho)= k \rho^\gamma, \ \ \frac65<\gamma<2.
\ee  
The linear stability of Lane-Emden stars is a classical topic in astrophysics~\cite{Chandrasekhar1938} and they also play an important role in our work as suitably rescaled limiting objects
in the Newtonian limit $\kappa\to0$, see Lemma~\ref{L:SMALLKAPPA}. For general (non-polytropic) equations of state, the stability analysis is considerably more complicated due to the absence of exact scaling invariance. In a recent work Lin and Zeng~\cite{LinZeng2020} showed that for a very general class of equations of state allowing for compact equilibria, essentially the same turning point principle as proposed by Wheeler applies. In fact, our strategy in this paper is based on analogous steps to~\cite{LinZeng2020}.
Of central importance in the proof of the turning point principle for the Euler-Poisson system is the 
the Newtonian limit of the operator $\Sigma_\kappa$  given by 
\begin{align}
\Sigma_0 : = -\frac1{4\pi}\Delta  - g_0', \ \ \Sigma_0:\dot H^1_r\to (\dot H^1_r)^\ast,
\end{align} 
where 
\begin{align}
g_0(r) & = k^{-\alpha}\left(\frac{\gamma-1}{\gamma}\right)^{\alpha} r^\alpha_+ \label{E:GZERODEF}\\
\alpha & := \frac1{\gamma-1}. \label{E:ALPHADEF} 
\end{align}
The subscript $+$ in $f_+$ refers to the positive part of the function $f$. 
Since $1<\gamma<2$ we have $\alpha>1$ and therefore $g_0$ is a $C^1$-function.
It is in particular shown in~\cite{LinZeng2020} that 
\begin{align}\label{E:SIGMAZEROPROPERTIES}
n^-(\Sigma_0)=1, \ \ \text{ker}(\Sigma_0)=\{0\}.
\end{align}
The operator $\Sigma_0$ can indeed be viewed as the Newtonian limit ($\kappa\to0^+$) of the sequence of operators $(\Sigma_\kappa)_{\kappa>0}$. This is a consequence of Lemma~\ref{L:SMALLKAPPA} and Corollary~\ref{C:APRIORI}.

\begin{remark}
In the context of the Euler-Poisson system, operator $\Sigma_0$ is  the reduced operator associated with the Lane-Emden steady states with equation of state $p = k\rho^\gamma$. 
\end{remark}

\begin{definition}[Winding index]\label{D:INDEX}
An important quantity in our analysis is the {\em winding index} $i_\kappa$:
\be\label{E:IKAPPADEF}
  i_\kappa = 
  \begin{cases}
    1&\text{if }  \frac{d}{d\kappa} M_\kappa \frac{d}{d\kappa}\left(\frac{M_\kappa}{R_\kappa}\right) >0 \ \ \text{ or } \ \frac{d}{d\kappa}M_\kappa=0, \\
    0&\text{if }  \frac{d}{d\kappa} M_\kappa\frac{d}{d\kappa}\left(\frac{M_\kappa}{R_\kappa}\right) <0 \ \ \text{ or } \ \frac{d}{d\kappa}\left(\frac{M_\kappa}{R_\kappa}\right)=0.
  \end{cases} 
  \ee
\end{definition}


\begin{remark}
It is shown in Lemma~\ref{L:NOCOEXISTENCE} that there is no $\kappa>0$ such that $\frac{d}{d\kappa}M_\kappa=\frac{d}{d\kappa}\left(\frac{M_\kappa}{R_\kappa}\right)=0$, so the winding index $i_\kappa$ is well-defined.
\end{remark}


\begin{theorem}\label{T:NEGATIVEMODES}
Consider the $1$-parameter family of solutions $(0,\infty)\ni \kappa\to(\rho_\kappa,\l_\kappa,\mu_\kappa)$ to the radially symmetric Einstein-Euler system. 
\begin{enumerate}
\item[(i)] 
The number of growing modes $n^u(\kappa)$ of the linearized EE-system around a steady state $(\rho_\kappa,\l_\kappa,\mu_\kappa)$ is given by the formula
\be\label{E:NEGATIVEMODES}
n^u(\kappa) = n^-(\Sigma_\kappa) - i_\kappa,
\ee
where $i_\kappa$ is given in Definition~\ref{D:INDEX} and $n^-(\Sigma_\kappa)$ is the negative Morse index of the operator $\Sigma_\kappa$.
\item[(ii)] 
The eigenvalues of the linearized system are discrete with finite multiplicity.
\end{enumerate}
\end{theorem}

\begin{remark}
The discreteness of the spectrum can also be obtained using Sturm-Liouville type methods. The formulation can be essentially read off from Chandrasekhar's pioneering work~\cite{Chandrasekhar1964}, for mathematically rigorous treatment see for example the work of Makino~\cite{Ma2016}. Our proof of discreteness in Theorem~\ref{T:NEGATIVEMODES} proceeds by a different method and capitalizes crucially on the separable Hamiltonian structure of the linearized operator. The same strategy has been used in the Euler-Poisson case~\cite{LinZeng2020} and it is a generally applicable procedure to other systems enjoying the separable Hamiltonian structure. 
\end{remark}


\begin{theorem}\label{T:TPP}
Consider the $1$-parameter family of solutions $(0,\infty)\ni \kappa\to(\rho_\kappa,\l_\kappa,\mu_\kappa)$ to the radially symmetric Einstein-Euler system. 
\begin{enumerate}
\item
{\em Turning Point Principle.}
The number of growing modes $n^u(\kappa)$ can only change at the extrema of the mass function $\kappa\to M_\kappa$.
At an extremum of $\kappa\to M_\kappa$, $n^u(\kappa)$ increases by 1 if the sign of $\frac{d}{d\kappa}M_\kappa\frac{d}{d\kappa}R_\kappa$
changes from $-$ to $+$ as $\kappa$ increases, and similarly it decreases by $1$ if the sign of $\frac{d}{d\kappa}M_\kappa\frac{d}{d\kappa}R_\kappa$
changes from $+$ to $-$ as $\kappa$ increases. Geometrically this implies that we ``gain" a growing mode if the mass-radius curve bends counter-clockwise at the extremum of $\kappa\to M_\kappa$, and we ``lose" a growing mode if the mass-radius curve bends clockwise at the extremum of $\kappa\to M_\kappa$. Here the horizontal axis corresponds to the star radius.
\item The number of growing modes goes to infinity as $\kappa$ goes to infinity, i.e.
\be\label{E:VERYUNSTABLE}
\lim_{\kappa\to\infty} n^u(\kappa)=\infty.
\ee
\end{enumerate}
\end{theorem}


\begin{remark}
In the physics literature the onset of the ``higher and higher order instabilites"~\cite{Thorne1966} as $\kappa\to\infty$ for static stars with extremely dense cores was first pointed out by Dimitriev and Holin~\cite{DiHo1963} in 1963, as well as Harrison~\cite{Ha1965} and Wheeler~\cite{HaThWaWh}.
\end{remark}

\begin{remark}
Part (ii) of Theorem~\ref{T:TPP} is a (considerable) strengthening of 
a result in~\cite{HaLiRe}, where it was shown that for $\kappa\gg1$ sufficiently large we have $n^u(\kappa)\ge1$.
\end{remark}

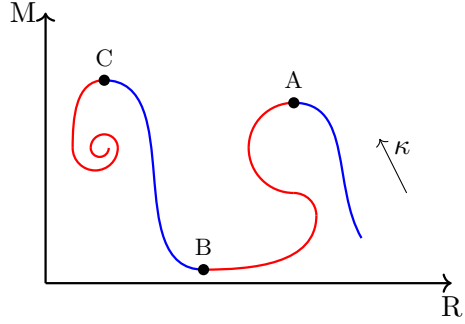
\begin{figure}
\begin{center}
\begin{tikzpicture}[domain=0:5, scale = 0.6]
\draw[->, black, line width = 0.30mm]   (0,0) -- (0,6);
\node at (-0.5,6){M};
\draw[->, black, line width = 0.30mm] (0,0) -- (9,0);
\node at (9,-0.5){R};
\draw[->] (8,2) -- (7.4,3.2);
\node at (7.9,3.0){$\kappa$};

\draw[blue, line width = 0.30mm] (7,1) to[out=120,in=0] (5.5,4);
\draw[red, line width = 0.30mm] (5.5,4) to[out=180,in=90] (4.5,3);
\draw[red, line width = 0.30mm] (4.5,3) to[out=270,in=180] (5.5,2);
\draw[red, line width = 0.30mm] (5.5,2) to[out=0,in=90] (6,1.5);
\draw[red, line width = 0.30mm] (6,1.5) to[out=270,in=0] (3.5,0.3);
\draw[blue, line width = 0.30mm] (3.5,0.3) to[out=180,in=0] (1.3,4.5);
\draw[red, line width = 0.30mm] (1.3,4.5) to[out=180,in=90] (0.6,3.0);
\draw[red, line width = 0.30mm] (1.6,3) arc(180:0:-0.5);
\draw[red, line width = 0.30mm] (1,3) arc(180:0:0.3);
\draw[red, line width = 0.30mm] (1.4,3) arc(180:0:-0.2);
\node at (5.5,4)[circle,fill,inner sep=1.5pt]{};
\node at (5.5,4.5){\footnotesize A};
\node at (3.5,0.3)[circle,fill,inner sep=1.5pt]{};
\node at (3.5,0.8){\footnotesize B};
\node at (1.3,4.5) [circle,fill,inner sep=1.5pt]{};
\node at (1.3,5) {\footnotesize C};
\end{tikzpicture}
\end{center}
\caption{Schematic sketch of a possible mass-radius curve based on the physics literature, see e.g.~\cite{HaThWaWh, Weinberg1972, Wald1984}. Blue portions correspond to spectrally stable equilibria. The first three local extrema of $\kappa\to M_\kappa$ are labelled by $A,B,C$. Starting with $\kappa\ll1$ small to the very right of the curve, equilibria are stable until we reach the first maximum of $M_\kappa$ at the point $A$ (the ``white dwarf" region). By Theorem~\ref{T:TPP} equilibria between $A$ and $B$ have 1 growing mode, the branch between $B$ and $C$ is again stable (the ``neutron star" region), and all the equilibria pass the point $C$ are unstable. The far left region of the graph features the infinite mass-radius spiral which leads to ever-increasing number of growing modes as it bends counterclockwise.}
\label{F:TPPEE}
\end{figure}


One of the central outcomes in the analysis of the Euler-Einstein system in~\cite{HaLiRe} is Theorem 5.20 which proves the existence 
and the associated exponential trichotomy 
decomposition of the phase space for the linearized flow~\eqref{E:FIRSTORDERFLOW}. Theorem 5.20 in~\cite{HaLiRe} is however not complete, as
it does not address the steady states whose central redshifts satisfy the nongeneric condition 
\be\label{E:NONGENERIC}
\frac{d}{d\kappa}M_\kappa \frac{d}{d\kappa}\left(\frac{M_\kappa}{R_\kappa}\right) =0.
\ee
For readers' convenience we state the complete version of the theorem, however, we only briefly sketch the proof in Section~\ref{S:PROOFS} as it follows closely the arguments in~\cite{HaLiRe,LinZeng2020}.


\begin{theorem}[Exponential trichotomy]\label{T:TRICHOTOMY}
Let the equation of state $\rho\to P(\rho)$ satisfy assumptions (P1)--(P4) and let $(\rho_\kappa,\mu_\kappa,\l_\kappa)$ be the 1-parameter family of steady states given 
by Proposition~\ref{P:EXISTENCE}. Then for any $\kappa>0$
the operator $J^{\kappa }\mathcal L^{\kappa }$ generates a $%
C^{0}$ group $e^{tJ^{\kappa }\mathcal L^{\kappa }}$ of bounded linear operators on $X_\kappa\times Y_\kappa$
and there exists a decomposition%
\begin{equation*}
X_\kappa\times Y_\kappa=E^{u}\oplus E^{c}\oplus E^{s},\quad
\end{equation*}
with the following properties:
\begin{itemize}
\item[(i)]
  $E^{u}\left( E^{s}\right) $ consist only of eigenvectors corresponding to
  negative (positive) eigenvalues of $J^{\kappa }\mathcal L^{\kappa }$ and  
  \begin{equation}
    \dim E^{u}=\dim E^{s}=n^{-}\left( \Sigma_{\kappa }\right) -i_{\kappa },
    \label{unstable-dimension-formula-EE}
  \end{equation}
  where $i_{\kappa }$ is defined by~\eqref{E:IKAPPADEF}. 
\item[(ii)]
  The quadratic form $\left( \mathcal L^{\kappa }\cdot ,\cdot \right)_X $
  vanishes on $E^{u,s}$, but is non-degenerate on $E^{u}\oplus E^{s} $, and 
  \begin{equation*}
    E^{c}=\left\{
    \begin{pmatrix} \rho \\ v \end{pmatrix}
    \in X_\kappa\times Y_\kappa \mid \left( \mathcal L^{\kappa }
    \begin{pmatrix} \rho \\ v \end{pmatrix},
    \begin{pmatrix} \rho_1 \\ v_1 \end{pmatrix}
    \right)_X =0\ \mbox{for all}\
    \begin{pmatrix} \rho_1 \\ v_1 \end{pmatrix}
    \in E^{s}\oplus E^{u}\right\}.
  \end{equation*}
\item[(iii)]
  $E^{c}, E^{u}, E^{s}$ are invariant under $e^{tJ^{\kappa }\mathcal L^{\kappa}}$.
\item[(iv)]
 Let $\lambda _{u}=
  \min \{\lambda \mid \lambda \in \sigma (J^{\kappa}L^{\kappa }|_{E^{u}})\}>0$.
  Then there exist $M>0$ such that 
  \begin{equation}
    \begin{split}
      & \left\vert e^{tJ^{\kappa }L^{\kappa }}|_{E^{s}}\right\vert
      \leq M e^{-\lambda _{u}t},\ t\geq 0,\\
      & \left\vert e^{tJ^{\kappa }L^{\kappa }}|_{E^{u}}\right\vert
      \leq M e^{\lambda _{u}t},\ t\leq 0,
    \end{split}
    \label{estimate-stable-unstable-EE}
  \end{equation}
  \begin{equation}
    \left\vert e^{tJ^{\kappa }L^{\kappa }}|_{E^{c}}\right\vert
    \leq M (1+|t|)^{k_{0}},
    \ t\in \mathbb{R},  \label{estimate-center-EE}
  \end{equation}
  where
  \begin{equation}
    k_{0}\leq 2.  \label{bound-k-0-EE}
  \end{equation}
 In the generic case $\frac{d}{d\kappa}M_\kappa\neq0$, we have $k_0=0$ and therefore the flow is Lyapunov stable on the center space $E^c$.
  \end{itemize}
\end{theorem}

\begin{remark}
Invariant subspaces and the exponential trichotomy are important for a refined description of the dynamics in the vicinity of the equilibria. Our result is closely related to the criticality picture emerging in the description of contrasting dynamics near nontrivial steady states, which is largely based on numerical and heuristic arguments~\cite{GaGu}. In the context of neutron stars, Noble and Choptuik~\cite{NoCh2016} numerically probed the dynamics near the unstable equilibria of the Einstein-Euler system, using the initial velocity and the central density (or equivalently $\kappa$) to parametrize their perturbations. The resulting dynamic picture is very rich, and leads to collapsing, dispersive, and time-periodic solutions with data starting out close to unstable equilibria.
\end{remark}

As already explained, the statement of Theorem~\ref{T:TPP} goes back to Zel'dovich~\cite{Ze1963} and Wheeler~\cite{HaThWaWh}, see also Section 10.11 of the book by Zel'dovich and Novikov~\cite{ZeNo}. It is also referred to as the {\em static criterion} or the {\em static approach}~\cite{BiKo2004,Tassoul1978} which, as formulated in the original work of Zel'dovich~\cite{Ze1963}, asserts that a growing mode is gained or lost at the extrema of the $\kappa\to M(\kappa)$ curve - specifically only at the maxima and minima and at no other extrema~\cite{Tassoul1978}. The word ``static" is used, as the stability can be read off from the location of the equilibrium on the mass-radius curve, which are natural astrophysical observables; in the process we avoid potentially cumbersome eigenvalue computations~\cite{Thorne1966}. Our formulation of this principle follows closely the one in~\cite{HaThWaWh}. In 1965 Thorne~\cite{Thorne1966} gave a more precise version of Wheeler's Turning Point Principle, and provided heuristic arguments for the main conclusion of part (ii) of Theorem~\ref{T:TPP}. In 1970 Calamai~\cite{Calamai1970} similarly gave a more refined argument for the static approach to stability. Various heuristic treatments of the ``static approach" can be found in the textbook by Shapiro and Teukolsky~\cite{ShTe1983} and Straumann~\cite{Straumann2013}.

The first comprehensive treatments of the (linear) stability study of the isentropic relativistic dynamic equilibria (stars) started with the seminal contribution of Chandrasekhar~\cite{Chandrasekhar1964}, which after the pioneering work of Oppenheimer and Volkov from 1939~\cite{OpVo} gave a big boost to the study of dynamic stability properties of stars. For a historical overview we refer the reader to the summer school notes of Thorne~\cite{Thorne1966} and the review paper of Bisnovatyi-Kogan~\cite{BiKo2004}. 
Chandrasekhar~\cite{Chandrasekhar1964} linearized the problem in the co-moving coordinates and formulated the spectral stability problem in terms of a suitable Rayleigh-Ritz minimization principle for the eigenvalues of the linearized operator. An alternative, purely ``Eulerean" characterization of spectral stability was derived by Thorne in~\cite{HaThWaWh} in terms of the second variation of the ADM-mass $M$ under the constraint of constant total particle number $N$. For more details we point the reader to~\cite{HaLiRe} and references therein. 
At the same time as Wheeler's work on turning point principle~\cite{HaThWaWh} Bardeen~\cite{Ba1965} proposed a slightly different turning point principle for so-called hot stars (where the thermodynamic temperature is not zero), also relying on the $M(R)$-diagram. A nice overview is given by Bardeen, Thorne, and Meltzer~\cite{BaThMe}, where both the spectral stability of a single star, as well as their behaviour along the mass-radius curve is discussed.

When studying the stability of self-gravitating systems, a distinction is made between the {\em dynamic} stability/instability - studied in this paper - and the {\em thermodynamic} stability/instability, see the work of Green, Schiffrin, and Wald~\cite{GrSchWa} for an extensive discussion. The latter instability sets in when an energy-like quantity - typically the {\em entropy} - can be infinitesimally increased with perturbations that keep other relevant conserved quantities infinitesimally zero.  This notion of stability is in general not equivalent to dynamic stability, but one can often formulate turning point principles along 1-parameter family of equilibria where entropy, or a {\em binding energy} is plotted against some other relevant conserved quantity.  A general criterion for determining turning point instabilities in this context was given by Sorkin~\cite{So1981,So1982}, which was later applied to the study of thermodynamic (in)stability of axisymmetric stars by Friedman, Ipser, and Sorkin~\cite{FrIpSo1988}. More recently, thermodynamic stability of radial and axisymmetric equilibria of  the Einstein-Euler system was investigated by Schiffrin and Wald~\cite{SchWa2013}, Roupas~\cite{Roupas2013}, both works containing a number of references on the topic. 



Turning point principles play an important role in the study of other relativistic self-gravitating systems. An important open problem in this context is the stability of radially symmetric galaxies, which are equilibria of the asymptotically flat Einstein-Vlasov system. Going back to Zel'dovich and Podurets~\cite{ZePo}, it is conjectured and numerically verified (see also Zel'dovich and Novikov~\cite{ZeNo}, and  more recent numerical investiagtion by Andr\'{e}asson and Rein~\cite{AnRe1}) that the stability of suitable 1-parameter families  exhibits a single exchange of stability to instability at some critical value of central redshift $\kappa=\kappa_{\text{max}}$. At $\kappa_{\text{max}}$ the so-called fractional binding energy has a maximum and for $\kappa>\kappa_{\text{max}}$ the equilibria are dynamically unstable. This stability scenario is very different from the mass-radius turning point principle that we prove in Theorem~\ref{T:TPP}, as in the case of stars stability can in principle be exchanged to instability, and then back to stability~\cite{HaThWaWh,Lindblom1997}, \bcr see Figure~\ref{F:TPPEE}\ec.  The works~\cite{HaRe2015, HaLiRe} show that the steady states are spectrally stable for small values of $\kappa$ and spectrally unstable for large values of $\kappa$ respectively, which is consistent with the Zeldovitch-Podurets stability picture. The ``large central redshift" instability is driven by the existence of an exponentially growing mode. To prove the existence of the growing mode and understand the invariant subspaces requires the full power of the separable Hamilton formulation of the Einsten-Vlasov system~\cite{HaLiRe,LinZengMemoirs}, as variational principles are inadequate for this purpose in the context of the Vlasov theory.
We also mention that related binding energy criteria play a role in the study of the stability of so-called boson stars~\cite{LeePang1989} as well as black holes/black rings in higher dimensions~\cite{FiMuRe, ArLo2005,SchWa2013}.

In the Newtonian context, we already mentioned that gaseous stars radial equilibria are embedded in 1-parameter families of the gravitational Euler-Poisson system~\eqref{E:CONTINUITYEP}--\eqref{E:POISSON}. On the other hand, the Newtonian limit of the Einstein-Vlasov system is the gravitational Vlasov-Poisson system and also admits 1-parameter families of radially symmetric equilibria, i.e. steady galaxies, for a given {\em microscopic} equation of state. While the Zel'dovich/Wheeler turning point principle was shown to be true in the macroscopic Euler-Poisson case~\cite{LinZeng2020}, such a principle is {\em wrong} for the Vlasov-Poisson case. To illustrate this, the well-known 1-parameter family of King solutions of the Vlasov-Poisson system possesses a mass-radius graph which spirals in to some asymptotic value $(M_\infty, R_\infty)$ with infinitely many winding points~\cite{RaRe2017}, but it is {\em nonlinearly dynamically stable} for any value of the central macroscopic density $\rho_0>0$~\cite{GuoRein2007,GuoLin2008,LeMeRa2011,LeMeRa2012}. The inadequacy of the mass-radius diagram to predict the offset of (linear) instability for the kinetic models such as Einstein-Vlasov and Vlasov-Poisson is intimately related to the more complicated Hamiltonian structure by comparison to their macroscopic (gaseous) counterparts. In particular, the space of dynamically accessible perturbations in the kinetic setting is infinite-codimensional, which is one of the reasons why an extension of our analysis in the present work to the Einstein-Vlasov system is difficult. 

We can however prove the radial equilibria of the Einstein-Vlasov system possess no growing modes for all $\kappa<\kappa_{\text{max}}$, see Theorem~\ref{T:TPPEV}. This is a consequence of Theorem~\ref{T:TPP} and the so-called {\em macro-micro} stability principle proved in Theorem 5.26 in~\cite{HaLiRe}. The latter  shows that, in a certain precise sense, the steady states of the Einstein-Vlasov system are ``more stable" than the steady states of the Einstein-Euler system. 


{\bf Plan of the paper.}
In Section~\ref{S:SPECTRAL} we prove a number of spectral properties of the linearized operator as it changes with the parameter $\kappa$.  Lemma~\ref{L:SMALLKAPPA} (unsurprisingly) shows that in the $\kappa\to0$ limit we recover the corresponding Newtonian problem, which is then used in conjunction with Lemma~\ref{L:LOCCONST} to compute both the kernel and the negative Morse index of $\Sigma_\kappa$ at small values of the parameter $\kappa>0$. This is used as a starting point for the continuity argument. Next, Lemma~\ref{L:KERNELLEMMA} gives a sharp characterization of the kernel of $\Sigma_\kappa$ for any value of $\kappa>0$ in terms of the critical points of the map $\kappa\to\frac{M_\kappa}{R_\kappa}$. The third most relevant result of Section~\ref{S:SPECTRAL} is the ``jump-lemma" formulated in Lemma~\ref{L:JUMPLEMMA}, showing that the negative Morse index $n^-(\Sigma_\kappa)$ can jump only at the critical points of the map $\kappa\to\frac{M_\kappa}{R_\kappa}$ and that this jump is {\em equal} to the jump of the winding index $i_\kappa$, see Definition~\ref{D:INDEX}. Section~\ref{S:PROOFS} is devoted to the proofs of Theorems~\ref{T:NEGATIVEMODES} and~\ref{T:TPP}, building on the preparatory results from~Section~\ref{S:SPECTRAL}. In Section~\ref{S:EV} we state and provide the proof of a sufficient stability condition for the radial equilibria of the asymptotically flat Einstein-Vlasov system.




\section{Spectral analysis}\label{S:SPECTRAL}


We consider the scaling 
\begin{align}
& y_\kappa(r)  = \kappa \bar y_\kappa(\kappa^ar) = \kappa \bar y (s), \ \ s:= \kappa^a r,  \label{E:BASICSCALING} \\
& a  = \frac{\alpha-1}{2}, \label{E:ADEF}
\end{align}
where we recall the definition of $\alpha$~\eqref{E:ALPHADEF}.
Our goal is to derive an equation for $\bar y_\kappa$ in the regime $0<\kappa\ll1$. 
A simple scaling argument shows that 
\begin{align*}
m(r) 
& = \kappa^{\frac{3-\alpha}{2}} \bar m_\kappa(s),
\end{align*}
where
\[
\bar m_\kappa(s) = \int_0^s 4\pi \tilde z^2 \bar\rho_\kappa(\tilde z) \, d\tilde z.
\]
The rescaled density $\bar\rho_\kappa$ and the rescaled function $g_\kappa$ (recall~\eqref{E:GDEF}) are defined via
\begin{align}
\bar\rho_\kappa(s) =g_\kappa(\bar y_\kappa(s)) & : =  \kappa^{-\alpha}g(\kappa \bar y_\kappa(s)) \ \text{ i.e. }\label{E:BARRHO1}\\
\rho_\kappa(r) & = \kappa^\alpha \bar\rho_\kappa(s). \label{E:BARRHO2}
\end{align}
We also introduce the rescaled pressure $\bar p_\kappa(s)$:
\begin{align}
\bar p_\kappa(s) & = \kappa^{-\alpha\gamma} P(\kappa^\alpha g_\kappa(\bar y_\kappa(s))) \ \text{ i.e. }\label{E:BARP1}\\
p_\kappa(r) & = \kappa^{\alpha\gamma} \bar p_\kappa(s). \label{E:BARP2}
\end{align}
Note that $p_\kappa(r) = P(\rho_\kappa(r))= \kappa^{\alpha\gamma} \bar p_\kappa(s)$, where $r$ and $s$ are related via the scaling~\eqref{E:BASICSCALING}--\eqref{E:ADEF}.

Plugging the above into~\eqref{E:YEQN}, 
we conclude that $\bar y_\kappa$ solves the following initial value problem
\begin{align}
\bar y_\kappa'(s) & = - \frac{1}{1-2 \kappa\bar m(s)/s} \left(\frac{\bar m(s)}{s^2} + 4 \pi \kappa s \bar p_\kappa(s)\right) \label{E:YKAPPA1}\\
\bar y_\kappa(0) & = 1.\label{E:YKAPPA2}
\end{align}

In the Newtonian limit we expect to recover the classical Lane-Emden stars. They are given as the unique solutions $[0,\infty)\ni r \to y_0(r)$ of the Cauchy problem:
\begin{align}
y'(s) & = - \frac{m_0(s)}{s^2} \label{E:NSS1}\\
y(0) & =1,\label{E:NSS2}
\end{align}
where
\begin{align}
m_0(s) = 4\pi \int_0^s z^2 \rho_0(z)\,dz, \ \ \rho_0(s) = g_0(y_0(s)), 
\end{align}
and $g_0$ is a $C^1$-function given by~\eqref{E:GZERODEF}.


\begin{lemma}\label{L:GGZERO}
There exists a $\kappa_0>0$ and a positive constant $C$ such that for all $y\in[0,\kappa_0]$ we have
\begin{align}
\lv g(y)-g_0(y)\rv + y\lv g'(y)-g_0'(y)\rv \le C y^{\alpha+1}.
\end{align}
\end{lemma}


\begin{proof}
By~\eqref{E:QDEF} and assumption (P2) we have for any $\rho\in[0,\zeta]$
\begin{align}
Q(\rho) & = \int_0^\rho \frac{P'(\sigma)}{\sigma + P(\sigma)}\,d\sigma \notag \\
& =  \int_0^\rho \frac{k\gamma \sigma^{\gamma-1}+ k(\sigma^{\gamma}f(\sigma))'}{\sigma + k\sigma^\gamma(1+f(\sigma))}\,d\sigma \notag \\
& = \int_0^\rho k \gamma \sigma^{\gamma-2} + O_{\sigma\to0}(\sigma^{2\gamma-3}) \,d\sigma \notag \\
& = k \frac{\gamma}{\gamma-1} \rho^{\gamma-1} + O_{\rho\to0}(\rho^{2\gamma-2}). \label{E:QASYMP}
\end{align}
Recall the definitions of $g$~\eqref{E:GDEF} and $g_0$~\eqref{E:GZERODEF}. It follows from~\eqref{E:QASYMP} that there exists a sufficiently small $\kappa_0$ such that 
\be\label{E:GEXPANSION}
g(y) = k^{-\alpha}\left(\frac{\gamma-1}{\gamma}\right)^{\alpha} y^\alpha + O_{y\to0^+}(y^{\alpha+1}) = g_0(y)+ O_{y\to0^+}(y^{\alpha+1}), \ y\in[0,\kappa_0],
\ee
and the claim follows.
\end{proof}


\begin{lemma}[The small redshift limit]\label{L:SMALLKAPPA}
There exist $\kappa^\ast, C>0$ such that for all $0\le \kappa<\kappa^\ast$ the following bound holds:
\begin{align}\label{E:BASICERRORBOUND}
\| \bar y_\kappa - y_0 \|_{C^1([0,\infty)} \le C \kappa,
\end{align}
where $y_0$ is the unique solution of~\eqref{E:NSS1}--\eqref{E:NSS2}.
\end{lemma}


\begin{proof}
This proof follows the ideas from~\cite{HaRe2015} where the small central redshift limit is investigated for the steady states of the Einstein-Vlasov system.
For any $\kappa>0$, let $\bar y_\kappa$ and $y_0$ be the unique global solutions to~\eqref{E:YKAPPA1}--\eqref{E:YKAPPA2} and~\eqref{E:NSS1}--\eqref{E:NSS2} defined on $[0,\infty)$. Functions $\bar y_\kappa, y_0$ are also strictly decreasing and by the choice of the initial condition $|\bar y_\kappa(s)|\le 1$, $|y_0(s)|\le 1$ for all $s\in[0,\infty)$.
We claim that there exists a $\kappa^\ast$ and a constant $C$ such that for all $0<\kappa\le \kappa^\ast$ we have the a priori bound
\begin{align}\label{E:APRIORI1}
\|\bar\rho_\kappa\|_{C^0([0,\infty))} + \|\bar p_\kappa\|_{C^0([0,\infty))} +\sup_{s\in[0,\infty)} \frac1{\lv 1-\frac{2\kappa\bar m_\kappa(s)}{s}\rv} \le C, \ \ \kappa\in(0,\kappa^\ast]
\end{align}
Bound for $\bar\rho_\kappa$ in~\eqref{E:APRIORI1} follows from two observations. According to~\eqref{E:BARRHO1} $\bar\rho_\kappa'(s) = \kappa^{1-\alpha}g'(\kappa \bar y_\kappa(s)) \bar y_\kappa'(s)\le 0$ for all $s\ge0$ and therefore $\bar\rho_\kappa(s)\le\bar\rho_\kappa(0) = \kappa^{-\alpha} g(\kappa)$. However, by~\eqref{E:GEXPANSION} $\kappa^{-\alpha} g(\kappa)\lesssim1$ for $\kappa$ sufficiently small and the claim follows. The same argument applies to $\bar p_\kappa$ due to~\eqref{E:BARP1} and assumption (P1). Finally, by the classical Buchdahl inequality for the spherically symmetric static solutions of the Euler-Einstein system we have 
$\sup_{r\in[0,\infty)}\frac{2m(r)}{r} = \sup_{s\in[0,\infty)}\frac{2\kappa\bar m_\kappa(s)}{s} \le \frac89$, which completes the proof of~\eqref{E:APRIORI1}.

 Since the Lane-Emden steady state associated with $y_0$ has a compact support whose extent corresponds to the unique zero of $s\mapsto y_0(s)$, there exists an $S_0>0$ such that $y_0(S_0)<0$ and $S_0$ is strictly to the right of the support of $\rho_0$.  It is then clear that for some constant $C$ (depending on $S_0$) we have the bound
\begin{align}\label{E:APRIORI2}
\sup_{s\in[0,S_0]}|\bar m_\kappa(s)| \le C, \ \ \kappa\in(0, \kappa^\ast].
\end{align}

It then follows that
\begin{align}
\lv\bar y_\kappa'(s) - y_0'(s)\rv & \le \lv \frac{4\pi \kappa s \bar p_\kappa (s)}{1-2 \kappa\bar m_\kappa(s)/s} \rv + 
\lv \frac1{1-2 \kappa\bar m_\kappa(s)/s}-1 \rv \frac{\bar m_\kappa(s)}{s^2} \notag \\
& \ \ \ \  + \lv \frac{\bar m_\kappa(s)-m_0(s)}{s^2}\rv \notag \\
& \le C\kappa + C\int_0^s \lv \bar\rho_\kappa(\sigma) - \rho_0(\sigma)\rv \,d\sigma \label{E:YKAPPAPRIMEBOUND}
\end{align}
Since $\bar\rho_\kappa - \rho_0 =\left(g_\kappa(\bar y_\kappa) - g_0(y_\kappa)\right) + \left(g_0(y_\kappa)-g_0(y_0)\right)$, we have
\begin{align}\label{E:RHOKAPPABOUND}
\lv\bar\rho_\kappa - \rho_0 \rv \le \lv g_\kappa(\bar y_\kappa) - g_0(\bar y_\kappa) \rv + C \lv \bar y_\kappa - y_0\rv,
\end{align}
where we have used that $g_0$~\eqref{E:GZERODEF} is a $C^1$ function with a uniformly bounded derivative on $(-\infty,1]$.
Let now $\kappa\le\kappa_0$, where $\kappa_0$ is defined in Lemma~\ref{L:GGZERO}.
Then for any $y\in[0,1]$
\begin{align}
g_\kappa(y) - g_0(y) 
& = \kappa^{-\alpha} g(\kappa y) - k^{-1} C_\gamma^{-\alpha} y^\alpha \notag \\
& = \kappa^{-\alpha} \left(g_0(\kappa y)+ O_{\kappa\to0^+}(\kappa^{\alpha+1} y^{\alpha+1})\right) - g_0(y) \notag \\
& = O(\kappa), \ \ y\in[0,1]. \label{E:GKAPPABOUND}
\end{align}
Plugging~\eqref{E:GKAPPABOUND} into~\eqref{E:RHOKAPPABOUND} and~\eqref{E:YKAPPAPRIMEBOUND}, we conclude
\begin{align}\label{E:BARYKAPPAINTERIOR}
\lv\bar y_\kappa'(s) - y_0'(s)\rv \le C\kappa + C\int_0^s \lv \bar y_\kappa - y_0\rv \,d\sigma, \ \ s\in[0,S_0].
\end{align}
Applying the Gr\"onwall inequality we conclude 
\be\label{E:BOUND1}
\sup_{0\le s< S_0}\left\vert \bar y_\kappa(s) - y_0(s) \right\vert \le C \kappa, \ \ \kappa\in[0,\kappa_0].
\ee
Since $y_0(S_0)<0$, for sufficiently small $\kappa_0$ we also have $y_\kappa(S_0)<0$ for all $\kappa\in[0,\kappa_0]$. If $s>S_0$ and $\kappa\in[0,\kappa_0]$ we have $\bar\rho_\kappa(s)=0$ and thus the solution has to be Schwarzschild in the vacuum region. In particular $\bar m_\kappa(s) = \bar M_\kappa$ is constant for all $s\ge S_0$ and 
\be\label{E:BARYKAPPAINVACUUM}
\bar y_\kappa'(s) = - \frac1{1-\frac{2\kappa \bar M_\kappa}{s}}\frac{\bar M_\kappa}{s^2},
\ee
which leads to the explicit formula
\[
\bar y_\kappa(s) = \bar y_\kappa (S_0) + \frac1{2\kappa}\log\left(1-\frac{2\kappa\bar M_\kappa}{S_0}\right) -  \frac1{2\kappa}\log\left(1-\frac{2\kappa \bar M_\kappa}{s}\right).
\]
Moreover
\[
y(s) = y_0 + \frac{M_0}{S_0} - \frac{M_0}{s},
\]
where $M_0$ is the total mass of the Newtonian solution $y_0$. Since $|\bar M_\kappa-M_0|\lesssim \kappa$ and $\lv \bar y_\kappa(S_0)-y_0(S_0)\rv\lesssim\kappa$, it follows easily that 
\be\label{E:BOUND2}
\sup_{S_0\le s< \infty}\left\vert \bar y_\kappa(s) - y_0(s) \right\vert \le C \kappa, \ \ \kappa\in[0,\kappa_0].
\ee
This completes the $C^0$-bound of~\eqref{E:BASICERRORBOUND}. The $C^1$-bound is now a simple consequence of~\eqref{E:BARYKAPPAINTERIOR}--\eqref{E:BARYKAPPAINVACUUM}.
\end{proof}


Since $y_\kappa(r)=\mu_\kappa(R_\kappa)-\mu_\kappa(r)$, in light of~\eqref{E:BASICSCALING} it is natural to define $\bar\mu_\kappa$ through the relationship
\be\label{E:BARMUKAPPADEF}
\bar \mu_\kappa(s) = \frac1\kappa\mu_\kappa(R_\kappa) -\bar y_\kappa(s).
\ee
In particular, $\mu_\kappa'(r) = \kappa^{\frac{\alpha+1}{2}}\bar\mu_\kappa'(s)$.
We also define $\bar\l_\kappa(s)$ through
\begin{align}
\bar\lambda_\kappa(s) & = \frac1\kappa \l_\kappa(r). \label{E:BARLAMBDAKAPPADEF}
\end{align}
A simple corollary of Lemma~\ref{L:SMALLKAPPA} are the following a priori bounds.


\begin{corollary}\label{C:APRIORI}
There exists a $\kappa^\ast>0$ sufficiently small  and constants $C,S_0>0$ such that for all $0<\kappa\le\kappa^\ast$
\begin{align}
|R_\kappa| & \le S_0, \label{E:RKAPPABOUND} \\
\|g_\kappa-g_0\|_{C^1([0,1])} & \le C\kappa \label{E:GKAPPAGZEROBOUND}\\
\|\bar\rho_\kappa - \rho_0\|_{C^0([0,\infty))} + \|\bar p_\kappa - p_0\|_{C^0([0,\infty))} &\le C \kappa \label{E:PRHOKAPPABOUND}\\
\|\bar\l_\kappa\|_{C^0([0,\infty))} + \kappa^{1-\alpha}\|\bar\l_\kappa' + \bar\mu_\kappa'\|_{C^0([0,\infty))} & \le C  \label{E:METRICAPRIORI}
\end{align} 
\end{corollary}


\begin{proof}
Bound~\eqref{E:RKAPPABOUND} is obvious from the proof of Lemma~\ref{L:SMALLKAPPA}, as by construction $R_\kappa< S_0$ for $\kappa$ sufficiently small and $S_0$ as in proof of Lemma~\ref{L:SMALLKAPPA}. The $C^0$-part of~\eqref{E:GKAPPAGZEROBOUND} has already been established in~\eqref{E:GKAPPABOUND}. Using~\eqref{E:BARRHO1}, we have
\[
g_\kappa'(y) - g_0'(y) = \kappa^{1-\alpha}g'(\kappa y) - k^{-1}\alpha C_\gamma^{-\alpha} y_+^{\alpha-1}
\]
and the claim then follows from Lemma~\ref{L:GGZERO}. Bound~\eqref{E:PRHOKAPPABOUND} is a direct consequence of~\eqref{E:BARRHO1}--\eqref{E:BARP2}, Lemma~\ref{L:SMALLKAPPA}, and~\eqref{E:GKAPPAGZEROBOUND}. The first claim in~\eqref{E:METRICAPRIORI} follows from the formula
\begin{align}
\bar\l_\kappa(s) = -\frac1{2\kappa}\log\left(1-\frac{2\kappa\bar m_\kappa(s)}{s}\right),
\end{align}
and the uniform bound $\frac{\bar m_\kappa(s)}{s}\le C$, a consequence of~\eqref{E:APRIORI1} and the definition of $\bar m_\kappa$. Finally, it is well-known that 
$
\mu_\kappa'+\l_\kappa' = 4\pi r \left(\rho_\kappa+p_\kappa\right) 
$
or in the rescaled variables
\begin{align}
\bar\mu_\kappa'(s) + \bar\l_\kappa'(s) = 4\pi \kappa^{\alpha-1} s \left(\bar\rho_\kappa(s)+\kappa \bar p_\kappa(s)\right),
\end{align}
where we remind the reader that $\alpha=\frac1{\gamma-1}$. Together with~\eqref{E:APRIORI1}, finite extent of the star and the bound~\eqref{E:RKAPPABOUND}, we conclude the remaining claim in~\eqref{E:METRICAPRIORI}.
\end{proof}

\begin{lemma}\label{L:CORRECTION}
Let the equation of state $\rho\to P(\rho)$ satisfy assumptions (P1)--(P4) and let $(\rho_\kappa,\mu_\kappa,\l_\kappa)$ be the 1-parameter family of steady states given 
by Proposition~\ref{P:EXISTENCE}. 
Then for any $\kappa>0$ the following identity holds:
\begin{align}
e^{2\mu_\kappa(R_\kappa)} = 1 - \frac{2M_\kappa}{R_\kappa}.
\end{align}
In particular 
\[
\text{sign}(\frac{d}{d\kappa}(\mu_\kappa(R_\kappa)))= - \text{sign}\frac{d}{d\kappa}\left(\frac{M_\kappa}{R_\kappa}\right).
\]
\end{lemma}

\begin{proof}
Recall that for any $r\ge0$ $e^{-2\l_\kappa(r)} = 1 - \frac{2m(r)}{r}$, where $m(r)=\int_0^r 4\pi \rho_\kappa(s) s^2\,ds$. In particular, $M_\kappa=\int_0^\infty 4\pi s^2\rho_\kappa\,ds = m(R_\kappa)$. Since $e^{\mu_\kappa(r)+\l_\kappa(r)} = 1$ for all $r\ge R_\kappa$, it finally follows that 
\[
e^{2\mu_\kappa(R_\kappa)} = e^{-2\l_\kappa(R_\kappa)} = 1 - \frac{2m(R_\kappa)}{R_\kappa} = 1 - \frac{2M_\kappa}{R_\kappa}.
\]
\end{proof}


\begin{lemma}\label{L:KERNELLEMMA}
Let the equation of state $\rho\to P(\rho)$ satisfy assumptions (P1)--(P4) and let $(\rho_\kappa,\mu_\kappa,\l_\kappa)$ be the 1-parameter family of steady states given 
by Proposition~\ref{P:EXISTENCE}. Then $\frac{d}{d\kappa}\left(\frac{M(\kappa)}{R(\kappa)}\right)\neq 0$ if and only if $\text{ker}\,\Sigma_\kappa=\{0\}$.
\end{lemma}

\begin{proof}
We first show that if $y_\kappa$ is a solution to~\eqref{E:YEQN} and~\eqref{E:YINITIAL} then $\Sigma_\kappa v_\kappa=0$ where
$v_\kappa:=\frac{d}{d\kappa} y_\kappa$. 
Differentiating~\eqref{E:YEQN} with respect to $\kappa$ we find that $v_\kappa$ solves
\be\label{E:VEQN}
v_\kappa' = 2\frac{d}{d\kappa}\l_\kappa y_\kappa' - e^{2\l_\kappa}\left(\frac{\frac{d}{d\kappa}m_\kappa}{r^2} + 4\pi r \frac{d}{d\kappa}p_\kappa\right),
\ee
where we have used $e^{-2\l_\kappa} = 1-\frac{2 m_\kappa}{r}$. The latter also implies 
\be\label{E:DKAPPALAMBDA}
\frac{d}{d\kappa}\l_\kappa = \frac{e^{2\l_\kappa}\frac{d}{d\kappa}m_\kappa}{r}.
\ee
Since $g'(y_\kappa) = \frac{\rho_\kappa+p_\kappa}{P'(\rho_\kappa)}$ we have 
\[
\frac{d}{d\kappa}p_\kappa = P'(\rho_\kappa) g'(y_\kappa) v = (\rho_\kappa+p_\kappa) v_\kappa.
\]
Since 
\be\label{E:MULAMBDA}
4\pi r e^{2\l_\kappa}(\rho_\kappa+p_\kappa) = \mu_\kappa'+\l_\kappa' 
\ee
we conclude from the previous identity 
\be\label{E:IMP}
4\pi r e^{2\l_\kappa}\frac{d}{d\kappa}p_\kappa = ( \mu_\kappa'+\l_\kappa') v_\kappa.
\ee
Substituting~\eqref{E:IMP} into~\eqref{E:VEQN} and multiplying it by $e^{\mu_\kappa+\l_\kappa}$ we obtain
\begin{align*}
\left(e^{\mu_\kappa+\l_\kappa} v_\kappa\right)' = - e^{\mu_\kappa+3\l_\kappa}\frac{\frac{d}{d\kappa}m_\kappa}{r^2} (2r\mu_\kappa'+1),
\end{align*}
where we have used the identity $y_\kappa' = - \mu_\kappa'$ and~\eqref{E:DKAPPALAMBDA}. This in turn yields
\begin{align}\label{E:IMP1}
- \frac {r^2 e^{-\mu_\kappa-3\l_\kappa}}{2r\mu_\kappa'+1} \left(e^{\mu_\kappa+\l_\kappa} v_\kappa\right)'   = \frac{d}{d\kappa}m_\kappa.
\end{align}
Since 
\begin{align*}
\frac{d}{dr}\frac{d}{d\kappa}m_\kappa = 4\pi r^2 g'(y_\kappa) v_\kappa = 4\pi r^2 \frac{\rho_\kappa+p_\kappa}{P'(\rho_\kappa)} v_\kappa = \frac{e^{-\mu_\kappa}}{\Psi_\kappa}v_\kappa,
\end{align*}
the claim follows after differentiating~\eqref{E:IMP1} with respect to $r$ and multiplying it by $e^{\mu_\kappa+\l_\kappa}$.
Let $v\in \text{ker}\,\Sigma_\kappa$. Since $v_\kappa$ and $v$ satisfy the same second order homogeneous linear ODE on $[0,\infty)$ and $v'(0)=v_\kappa'(0)=0$, there exists a constant $C\neq0$ such that 
\[
v = C v_\kappa.
\]
However, since $v\in \dot H^1_r$ we have $\lim_{r\to\infty}v(r)=0$ and thus $\lim_{r\to\infty}v_\kappa(r)=0$. Since 
$v_\kappa(r) = \frac{d}{d\kappa}\left(\mu_\kappa(R_\kappa)\right) - \frac{d}{d\kappa}\mu_\kappa(r)$, we must have $\frac{d}{d\kappa}\left(\mu_\kappa(R_\kappa)\right)=0$,
which by Lemma~\ref{L:CORRECTION} gives $\frac{d}{d\kappa}\left(\frac{M_\kappa}{R_\kappa}\right)=0$.
\end{proof}


\begin{lemma}\label{L:LOCCONST}
Let the equation of state $\rho\to P(\rho)$ satisfy assumptions (P1)--(P4) and let $(\rho_\kappa,\mu_\kappa,\l_\kappa)$ be the 1-parameter family of steady states given 
by Proposition~\ref{P:EXISTENCE}. 
Then there exists a $\kappa^\ast>0$ such that for all $0<\kappa<\kappa^\ast$, $n^-(\Sigma_\kappa)=1$. Moreover, as a function of $\kappa$, $n^-(\Sigma_\kappa)$ is constant on any open interval not containing critical points of the map $\kappa\to\frac{M_\kappa}{R_\kappa}$.
\end{lemma}

\begin{proof}
If we set 
\[
\tilde\Sigma_\kappa : = e^{-\mu_\kappa-\l_\kappa} \Sigma_\kappa,
\]
then it is clear that $n^-(\Sigma_\kappa)=n^-(\tilde\Sigma_\kappa)$. Since 
\[
e^{-\mu_\kappa}\Psi_\kappa^{-1} = 
\begin{cases}
 \frac{\rho_\kappa+p_\kappa}{P'(\rho_\kappa)}, \ \ & \rho_\kappa>0 \\
 0, \ \ \ \ & \rho_\kappa=0,
\end{cases}
\]
it follows from~\eqref{E:QDEF}--\eqref{E:GDEF} that $e^{-\mu_\kappa}\Psi_\kappa^{-1}=g'(y_\kappa)$.
Therefore the operator
$\tilde\Sigma_\kappa$ for any $\phi\in \dot H^1$ reads
\begin{align*}
\tilde\Sigma_\kappa \phi & = 
-\frac{1}{4\pi r^2}
\frac{d}{dr}\left(\frac{e^{-\mu_\kappa-3\l_\kappa}r^2}{2r\mu_\kappa'+1}
\frac{d}{dr} \left(e^{\mu_\kappa+\l_\kappa}\phi \right)\right) -g'(y_\kappa)\phi \\
& = -\frac{1}{4\pi r^2}
\frac{d}{dr}\left(\frac{e^{-2\l_\kappa}r^2}{2r\mu_\kappa'+1}
\frac{d}{dr} \left(\left(\mu_\kappa'+\l_\kappa'\right)\phi + \phi' \right)\right) -g'(y_\kappa)\phi. 
\end{align*}
Using~\eqref{E:BARMUKAPPADEF}--\eqref{E:BARLAMBDAKAPPADEF}, the scaling~\eqref{E:BASICSCALING} and the identities~\eqref{E:BARRHO1}--\eqref{E:BARP2} we conclude
\begin{align*}
\tilde \Sigma_\kappa\phi = \kappa^{\alpha-1} \bar\Sigma_\kappa \bar\phi, \ \ \bar\phi(s) = \phi(r), \ \ \phi\in D(\tilde\Sigma_\kappa),
\end{align*}
where we recall the scaling~\eqref{E:BASICSCALING} and
\begin{align*}
 \bar\Sigma_\kappa \psi = - \frac1{4\pi s^2}\frac{d}{ds}\left(\frac{e^{-2\kappa\bar\l_\kappa}s^2}{2\kappa s\bar\mu_\kappa'(s)+1} 
 \left(\kappa(\bar\mu_\kappa'(s)+\bar\l_\kappa'(s))\psi + \frac{d}{ds}\psi \right)\right) 
 - g_\kappa'(\bar y_\kappa) \psi.
\end{align*}

We now proceed to obtain an upper bound for $\|\bar\Sigma_\kappa-\Sigma_0\|$. Integrating-by-parts it is easy to see that
\begin{align}
\langle \left(\bar\Sigma_\kappa-\Sigma_0\right)\psi\,,\,\psi\rangle
=& \int_0^\infty \left(\frac{e^{-2\kappa\bar\l_\kappa}}{2\kappa s\bar\mu_\kappa'(s)+1}-1\right)|\psi'(s)|^2 \,s^2ds \notag \\
& + \kappa \int_0^{R_\kappa} \frac{e^{-2\kappa\bar\l_\kappa}}{2\kappa s\bar\mu_\kappa'(s)+1}\left(\bar\mu_\kappa'(s)+\bar\l_\kappa'(s)\right) \psi(s) \psi'(s)  \,s^2ds \notag \\
& - 4\pi \int_0^\infty \left(g_\kappa'(\bar y_\kappa)-g_0'(y_0)\right)\psi(s)^2 \, s^2ds. \notag
\end{align}
Lemma~\ref{L:SMALLKAPPA} and Corollary~\ref{C:APRIORI} imply that for all $0<\kappa\le\kappa^\ast$
\begin{align}
\lv \langle \left(\bar\Sigma_\kappa-\Sigma_0\right)\psi\,,\,\psi\rangle \rv \le & C\kappa \int_0^\infty |\psi'(s)|^2 \,s^2 ds + C\kappa^\alpha \int_0^{S_0} |\psi(s)||\psi'(s)|\, s^2ds 
\notag \\
&+ C\kappa \int_0^{S_0} |\psi(s)|^2\, s^2ds, \label{E:QUADRFORMBOUND}
\end{align}
where in the last bound we have used that the supports of $g_\kappa\circ\bar y_\kappa$ and $g_0\circ y_0$ are both contained in $[0,S_0]$. Finally, since $\|\phi\|_{L^6(\mathbb R^3)}\lesssim \|\nabla\phi\|_{L^2(\mathbb R^3)}$ for any $\phi\in \dot H^1_r$ it follows from H\"older's inequality that 
$\|\phi\|_{L^2(B_{S_0}({\bf 0}))}\le \|\phi\|_{L^6(\mathbb R^3)} \le \|\nabla\phi\|_{L^2(\mathbb R^3)}$. Therefore, using Cauchy-Schwarz and~\eqref{E:QUADRFORMBOUND} we conclude
$
\lv \langle \left(\bar\Sigma_\kappa-\Sigma_0\right)\psi\,,\,\psi\rangle \rv \le C \kappa \|\psi\|_{\dot H^1_r}^2
$
for all $\psi\in \dot H^1_r$,
which in turn implies 
\[
\|\bar\Sigma_\kappa-\Sigma_0\| \lesssim \sqrt \kappa, \ \ 0<\kappa\le\kappa^\ast.
\]
By~\eqref{E:SIGMAZEROPROPERTIES} the operator $\Sigma_0$ is nondegenerate and therefore by Proposition 2.3 in~\cite{LinZeng2020} it follows that $n^-(\Sigma_\kappa)=n^-(\Sigma_0)=1$ for sufficiently small $\kappa$. By the same proposition, the value of $n^-(\Sigma_\kappa)$ can only change for those $\kappa>0$ where
the kernel of $\Sigma_\kappa$ is nontrivial, i.e. only at the critical points of $\kappa\to\frac{M_\kappa}{R_\kappa}.$

Strictly speaking, to apply Proposition 2.3 from~\cite{LinZeng2020} we must show that the operators $\Sigma_\kappa,\Sigma_0$ satisfy the assumption (G3) from~\cite{LinZeng2020}. By definition an operator $L: X\to X^\ast$ satisfies the property (G3) if 
it is bounded and self-dual and the Hilbert space $X$ can be decomposed into the direct sum of three closed subspaces
\be\label{E:DECOMP}
X=X_-\oplus \text{ker} L \oplus X_+, \ \ n^-(L):=\text{dim}(X_-)<\infty,
\ee
and moreover 1) $\langle Lu,u\rangle<0$ for all $u\in X_{-}\setminus\{0\}$ and 2) there exists a constant $\delta>0$ such that 
$\langle Lu,u\rangle\ge \delta \|u\|_X^2$ for all $u\in X_+$. The self-duality and boundedness of $\Sigma_\kappa$ is clear. 
To see that the decomposition~\eqref{E:DECOMP} holds we consider the operator $\tilde\Sigma_\kappa$ defined above.
We note that for any $\mu\in \dot H^1_r$ we have 
$\langle \Sigma_{\kappa }\mu ,\mu\rangle
=\langle \tilde\Sigma_{\kappa }\phi ,\phi \rangle$, where
$\phi =e^{\mu _{\kappa }+\lambda _{\kappa }}\mu$.
It suffices to show that 
\begin{equation}\label{positivity-sigma-kappa}
\langle \tilde\Sigma_{\kappa }\phi ,\phi \rangle \geq
\delta _{1}\left\Vert \nabla \phi \right\Vert_{L^2_r}^2, 
\end{equation}
with some  $\delta_{1}>0$ and for $\phi$
in a finite co-dimensional subspace of $\dot H^1_{r}$. 
This will in particular imply that the kernel and the space corresponding
to the negative part of the spectrum of the operator $\Sigma_\kappa$ are
at most finite-dimensional.  To prove~\eqref{positivity-sigma-kappa} we write 
\[
\tilde\Sigma_{\kappa }\phi = \mathscr S_0\phi+V_{\kappa } \phi, \ \ \mathscr S_0:=-e^{-\mu_\kappa-\l_\kappa}\Delta_\kappa,
\ \ V_\kappa = -g'(y_\kappa).
\]
\begin{equation*}
  \tilde\Sigma_{\kappa }
  =\mathscr S _{0}^{\frac{1}{2}}\left( \mathrm{id}+\mathscr S_{0}^{-\frac{1}{2}}
  V_{\kappa }\mathscr S _{0}^{-\frac{1}{2}}\right) \mathscr S_{0}^{\frac{1}{2}}
  =\mathscr S_{0}^{\frac{1}{2}}\tilde{\mathscr S}_{\kappa }\mathscr S_{0}^{\frac{1}{2}},
\end{equation*}
where 
\begin{equation*}
  \tilde{\mathscr S}_{\kappa }
  =\mathrm{id}+\mathscr S _{0}^{-\frac{1}{2}}V_{\kappa }\mathscr S_{0}^{-\frac{1}{2}}.
\end{equation*}
Since there exists $c_{0}>c_{1}>0$ such that 
\begin{equation*}
-c_{1}\Delta \leq \mathscr S _{0}\leq -c_{0}\Delta ,\quad \Delta =\frac{1}{r^{2}}
\frac{d}{dr}\left( r^{2}\frac{d}{dr}\right) ,
\end{equation*}
the operator $\mathscr S _{0}^{\frac{1}{2}}\colon \dot{H}_{1,r}\to L_{r}^{2}$
is an isomorphism. For $\phi \in \dot{H}_{1,r}$ we define
$\psi =\mathscr S_{0}^{\frac{1}{2}}\phi \in L_{r}^{2}$.
Since
$( \tilde\Sigma_\kappa \phi,\phi)_{L_r^2}
=( \tilde{\mathscr S}_{\kappa}\psi ,\psi)_{L_r^2}$,
the proof of (\ref{positivity-sigma-kappa}) is reduced to check that
$( \tilde{\mathscr S}_{\kappa }\psi ,\psi )_{L_r^2}$ is uniformly
positive for $\psi$ in a finite co-dimensional subspace of $L_{r}^{2}$.
We shall show that the operator 
\begin{equation*}
  \tilde{\mathscr S}_{\kappa }-\mathrm{id}
  =\mathscr S _{0}^{-\frac{1}{2}}V_{\kappa }\mathscr S
_{0}^{-\frac{1}{2}}:L_{r}^{2}\rightarrow L_{r}^{2}
\end{equation*}
is compact. Then it follows that the operator $\tilde{\mathscr S}_{\kappa }$ has
finite dimensional eigenspaces for negative and zero eigenvalues, and
$\tilde{\mathscr S}_{\kappa }$ is uniformly positive on the complement space. To
show the compactness of $\tilde{\mathscr S}_{\kappa }-\text{id}$, we take a sequence
$(\psi_{n}) \subset L_{r}^{2}$ such that $\psi_{n}\rightharpoonup 0$ weakly in
$L^{2}$ and show that
$\left\Vert \left( \tilde{\mathscr S}_{\kappa }-\mathrm{id}\right)
\psi _{n}\right\Vert _{L^{2}}\to 0$,
as $n\to \infty $. Indeed, by Hardy's inequality in Fourier space, 
\begin{eqnarray*}
  \left\Vert \left(\tilde{\mathscr S}_{\kappa}-\mathrm{id}\right)
  \psi_{n}\right\Vert_{L^{2}}
  &\leq&
  C \left\Vert \left( -\Delta \right)^{-\frac{1}{2}} V_{\kappa}
  \mathscr S _{0}^{-\frac{1}{2}}\psi _{n}\right\Vert _{L^{2}} 
  =
  C \left\Vert \frac{1}{|\xi|}
  \left( V_{\kappa }\mathscr S_{0}^{-\frac{1}{2}}\psi_{n}\right)^{\wedge}
  (\xi) \right\Vert _{L^{2}} \\
 & \leq&
  C \left\Vert |x| V_{\kappa }
  \mathscr S _{0}^{-\frac{1}{2}}\psi_{n}\right\Vert_{L^{2}} \to 0, \ \ \text{as $n\to\infty$},
\end{eqnarray*}
since $\mathscr S _{0}^{-\frac{1}{2}}\psi_{n}$ is bounded in $\dot{H}^1_{r}$
and $V_{\kappa }$ has compact support. This shows that $\Sigma_\kappa$ satisfies the property (G3) for
$\kappa>0$.
The same proof works for $\Sigma_0$.
\end{proof}


\begin{lemma}\label{L:NOCOEXISTENCE}
Let the equation of state $\rho\to P(\rho)$ satisfy assumptions (P1)--(P4) and let $(\rho_\kappa,\mu_\kappa,\l_\kappa)$ be the 1-parameter family of steady states given 
by Proposition~\ref{P:EXISTENCE}. 
Then there exists no $\kappa>0$ such that $\frac{d}{d\kappa}M_\kappa=\frac{d}{d\kappa}\left(\frac{M_\kappa}{R_\kappa}\right)=0$.
\end{lemma}


\begin{proof}
Assume the opposite, i.e. $\frac{d}{d\kappa}M_\kappa=\frac{d}{d\kappa}\left(\frac{M_\kappa}{R_\kappa}\right)=0$ for some $\kappa>0$. By Lemma~\ref{L:KERNELLEMMA}, we have $\Sigma_\kappa(\frac{d}{d\kappa}\mu_\kappa)=0$ which is equivalent to 
\be\label{E:NOCE1}
e^{-\mu_\kappa-\l_\kappa}\Delta_\kappa(\frac{d}{d\kappa}\mu_\kappa) - g'(y_\kappa)\frac{d}{d\kappa}\mu_\kappa=0, \ \ r\in[0,\infty).
\ee
Integrating the above relation over $B_\kappa$ and using $g'(y_\kappa)\frac{d}{d\kappa}\mu_\kappa = - \frac{d}{d\kappa}\rho_\kappa$, we obtain
\begin{align*}
\frac{e^{-\mu_\kappa-3\l_\kappa}}{2r\mu_\kappa'+1}\frac{d}{dr}\left(e^{\mu_\kappa+\l_\kappa}\frac{d}{d\kappa}\mu_\kappa\right)\Big|_{r=R_\kappa} = - \frac{d}{d\kappa}M_\kappa =0,
\end{align*}
where the last equality follows by our assumption.
Therefore $\frac{d}{dr}\left(e^{\mu_\kappa+\l_\kappa}\frac{d}{d\kappa}\mu_\kappa\right)\Big|_{r=R_\kappa}=0$, and thus from~\eqref{E:NOCE1} we conclude
$\frac{d}{dr}\left(e^{\mu_\kappa+\l_\kappa}\frac{d}{d\kappa}\mu_\kappa\right)=0$ for all $r\ge R_\kappa$. Since $\frac{d}{d\kappa}\mu_\kappa$ vanishes at asymptotic infinity, we conclude $\frac{d}{d\kappa}\mu_\kappa=0$ for all $r\ge R_\kappa$. Again by~\eqref{E:NOCE1} we have $\frac{d}{d\kappa}\mu_\kappa=0$ for all $r\ge 0$, which is clearly a contradiction, since $\frac{d}{d\kappa}\mu_\kappa\Big|_{r=0}=-\frac{d}{d\kappa}y_\kappa\Big|_{r=0}=-1$.
\end{proof}




\begin{lemma}\label{L:JUMPLEMMA}
Let the equation of state $\rho\to P(\rho)$ satisfy assumptions (P1)--(P4) and let $(\rho_\kappa,\mu_\kappa,\l_\kappa)$ be the 1-parameter family of steady states given 
by Proposition~\ref{P:EXISTENCE}. 
Let $\bar\kappa>0$ be an isolated critical point of the map $\kappa\to\frac{M_\kappa}{R_\kappa}$. Then 
\begin{align}\label{E:JUMPFORMULA}
n^-(\Sigma_{\bar\kappa+})-n^-(\Sigma_{\bar\kappa-}) = i_{\bar\kappa+}- i_{\bar\kappa-}. 
\end{align}
In other words the jump of $n^-(\Sigma_\kappa)$ equals the jump in $i_\kappa$.
\end{lemma}


\begin{proof}
Analogously to~\cite{LinZeng2020} we define the operator
\begin{align*}
\tilde{\mathbb L}_\kappa = K_\kappa^{-1} \mathbb L_\kappa K_\kappa.
\end{align*}
Locally around $\kappa=\bar\kappa$ there exists a curve $\ell(\kappa)$ of eigenvalues of $\tilde{\mathbb L}_\kappa$ such that $\ell(\bar\kappa)=0$. The associated eigenvalues are normalized so that
\be\label{E:NORMALISATION}
\left\| h_\kappa - \frac{d}{d\kappa}\rho_\kappa\Big|_{\kappa=\bar\kappa}\right\| \to 0, \ \ \text{ as $\kappa\to\bar\kappa$}.
\ee
We denote $\tilde h_\kappa = K_\kappa h_\kappa$, so that $\mathbb L_\kappa \tilde h_\kappa = \ell(\kappa) \tilde h_\kappa$. As a consequence,
\begin{align}
\ell(\kappa)\left(\tilde h_\kappa, \frac{d}{d\kappa}\rho_\kappa\right)_{X_\kappa} 
& = \left(\mathbb L_{\kappa}\tilde h_\kappa, \frac{d}{d\kappa}\rho_\kappa\right)_{X_\kappa} 
 =  \left\langle L_\kappa \tilde h_\kappa , \frac{d}{d\kappa}\rho_\kappa \right\rangle 
 =  \left\langle \tilde h_\kappa ,  L_\kappa \left(\frac{d}{d\kappa}\rho_\kappa\right) \right\rangle. \label{E:IMPID}
\end{align}
By~\eqref{E:FUNDAMENTAL} $L_\kappa \frac{d}{d\kappa}\rho_\kappa = e^{\mu_\kappa}\Psi_\kappa \Sigma_\kappa \left(\frac{d}{d\kappa}\rho_\kappa\right)$.
On the other hand,
\[
\Delta_\kappa \bar\mu_{\frac{d}{d\kappa}\rho_\kappa}= e^{\mu_\kappa+\l_\kappa}\frac{d}{d\kappa}\rho_\kappa
= e^{\mu_\kappa+\l_\kappa} g'(y_\kappa) \frac{d}{d\kappa} y_\kappa = e^{\l_\kappa}\Psi_\kappa^{-1}\frac{d}{d\kappa} y_\kappa.
\]
Since by the proof of Lemma~\ref{L:KERNELLEMMA} $e^{\l_\kappa}\Psi_\kappa^{-1}\frac{d}{d\kappa} y_\kappa = - \Delta_\kappa\frac{d}{d\kappa} y_\kappa$,
we conclude $\Delta_\kappa\left(\frac{d}{d\kappa} \rho_\kappa+\frac{d}{d\kappa} y_\kappa\right)=0$. This readily implies that there exists a constant $C(\kappa)$ such that 
\[
\frac{d}{d\kappa} \rho_\kappa+\frac{d}{d\kappa} y_\kappa = C(\kappa) e^{-\mu_\kappa-\l_\kappa}, \ \ r\ge0.
\]
Letting $r\to\infty$, we obtain
\[
C(\kappa) = \lim_{r\to\infty}\frac{d}{d\kappa} y_\kappa(r) = \frac{d}{d\kappa}\left[\mu_\kappa(R_\kappa)\right]
- \lim_{r\to\infty}\frac{d}{d\kappa}\mu_\kappa(r) = \frac{d}{d\kappa}\left[\mu_\kappa(R_\kappa)\right]
\]

In particular, 
\begin{align}
L_\kappa\left(\frac{d}{d\kappa} \rho_\kappa\right) & = e^{\mu_\kappa}\Psi_\kappa \Sigma_\kappa\left(\frac{d}{d\kappa} \rho_\kappa\right) \notag   \\
& = e^{\mu_\kappa}\Psi_\kappa \Sigma_\kappa \left(C(\kappa) e^{-\mu_\kappa-\l_\kappa} -\frac{d}{d\kappa} y_\kappa \right) \notag  \\
& = e^{\mu_\kappa}\Psi_\kappa \,  e^{\l_\kappa}\Psi_\kappa^{-1} \, C(\kappa) e^{-\mu_\kappa-\l_\kappa} \notag \\
& = C(\kappa) = \frac{d}{d\kappa}\left[\mu_\kappa(R_\kappa)\right], \label{E:LKDKAPPARHOKAPPA}
\end{align}
where we have used the identities $\Sigma_\kappa\left(\frac{d}{d\kappa} y_\kappa\right)= \Delta_\kappa \left(C(\kappa) e^{-\mu_\kappa-\l_\kappa}\right)=0$.
Using this in~\eqref{E:IMPID} we conclude
\begin{align}
\frac{\ell(\kappa)}{\frac{d}{d\kappa}\left[\mu_\kappa(R_\kappa)\right]}=   
\frac{\int_{B_\kappa} \tilde h_\kappa \,dx}{\left(\tilde h_\kappa, \frac{d}{d\kappa}\rho_\kappa\right)_{X_\kappa}}. \label{E:SAMESIGN}
\end{align}
Letting $\kappa\to\bar\kappa$ and bearing in mind the normalization~\eqref{E:NORMALISATION}, the right-hand side of the above identity converges to
\begin{align}\label{E:SAMESIGN2}
\frac{\int_{B_\kappa}\frac{d}{d\kappa}\rho_\kappa \,dx}{ \|\frac{d}{d\kappa}\rho_\kappa\|_{X_\kappa}^2 } \Big|_{\kappa=\bar\kappa}
= \frac{ \frac{d}{d\kappa} M_\kappa}{ \|\frac{d}{d\kappa}\rho_\kappa\|_{X_\kappa}^2 } \Big|_{\kappa=\bar\kappa}.
\end{align}
Since by Lemma~\ref{L:NOCOEXISTENCE} $\frac{d}{d\kappa}M_\kappa\Big|_{\kappa=\bar\kappa}\neq0$, the sign of $\frac{d}{d\kappa}M_\kappa\Big|_{\kappa=\bar\kappa}$ is constant in a small neighbourhood of $\bar\kappa$. Since by Lemma~\ref{L:CORRECTION} $\text{sgn}\frac{d}{d\kappa}\left[\mu_\kappa(R_\kappa)\right]
=-\text{sgn}\frac{d}{d\kappa}\left(\frac{M_\kappa}{R_\kappa}\right)$ we conclude from~\eqref{E:SAMESIGN} and~\eqref{E:SAMESIGN2} that
\[
n^-(L_{\bar\kappa+})- n^-(L_{\bar\kappa-}) = i_{\bar\kappa+}-i_{\bar\kappa-}
\] 
as desired, see Definition~\ref{D:INDEX}.
\end{proof}

\section{Proofs of the main theorems}\label{S:PROOFS}

Proofs of Theorems~\ref{T:NEGATIVEMODES} and~\ref{T:TPP} follow closely the structure of proofs of Theorems 1.1 and 1.2 in~\cite{LinZeng2020}. 

{\bf Proof of Theorem~\ref{T:NEGATIVEMODES}.}
We first recall that the number of unstable modes $n^u(\kappa)$ equals to $n^-(L_\kappa\big|_{\overline{R(A_\kappa)}})$, where
the space of dynamically accessible perturbations $\overline{R(A_\kappa)}\subset X_\kappa$ is explicitly described in~\eqref{E:DYNACC}.
{\em Proof of part (i).} \noindent
{\em Case 1.} let $\frac{d}{d\kappa}\left(\frac{M_\kappa}{R_\kappa}\right)\neq 0$. It then follows from~\eqref{E:LKDKAPPARHOKAPPA}
and Lemma~\ref{L:CORRECTION} that 
\begin{align}
\langle L_\kappa \frac{d\rho_\kappa}{d\kappa}, \frac{d\rho_\kappa}{d\kappa} \rangle
& = - \frac1{1-\frac{2M_\kappa}{R_\kappa}} \frac{d}{d\kappa}\left(\frac{M_\kappa}{R_\kappa}\right)
\int_{B_\kappa} \frac{d\rho_\kappa}{d\kappa}\,dx  \notag \\
& = - \frac1{1-\frac{2M_\kappa}{R_\kappa}} \frac{d}{d\kappa}\left(\frac{M_\kappa}{R_\kappa}\right) \frac{d}{d\kappa}M_\kappa 
\label{E:KEYFORMULA}
\end{align}
Moreover $\rho\in \overline{R(A_\kappa)}$ if and only if $\langle L_\kappa \frac{d\rho_\kappa}{d\kappa}, \rho \rangle=0$.
\begin{enumerate}
\item[Case 1a).] Let $\frac{d}{d\kappa}M_\kappa\neq 0$. Then it is clear from~\eqref{E:KEYFORMULA} that 
\begin{align*}
n^-(L_\kappa\big|_{\overline{R(A_\kappa)}})
& = \begin{cases}
n^-(L_\kappa) -1 & \ \text{ if } \ \frac{d}{d\kappa}\left(\frac{M_\kappa}{R_\kappa}\right) \frac{d}{d\kappa}M_\kappa<0 \\
n^-(L_\kappa)  & \ \text{ if } \ \frac{d}{d\kappa}\left(\frac{M_\kappa}{R_\kappa}\right) \frac{d}{d\kappa}M_\kappa>0 \\
\end{cases} \notag \\
& = n^-(\Sigma_\kappa)-i_\kappa.
\end{align*}
\item[Case 1b).] Let now $\frac{d}{d\kappa}M_\kappa=0$. From~\eqref{E:KEYFORMULA} we conclude
\begin{align}
\langle L_\kappa \frac{d\rho_\kappa}{d\kappa}, \frac{d\rho_\kappa}{d\kappa} \rangle =0, \ \ 
\frac{d\rho_\kappa}{d\kappa} \in \overline{R(A_\kappa)}. \label{E:CASE1B}
\end{align}
Since $\frac{d}{d\kappa}M_\kappa\neq 0$ by Lemma~\ref{L:NOCOEXISTENCE} $\text{ker} L_\kappa=\{0\}$.
Choose $\psi\in X_\kappa\setminus  \overline{R(A_\kappa)}$ and normalize it so that $\int_{B_\kappa}\psi\,dx=1$.
Consider the subspaces $S_0,S_1\subset X_\kappa$ defined by 
\begin{align*}
S_0 := \text{span}\left\{\psi,\frac{d\rho_\kappa}{d\kappa}\right\}, \ \ 
S_1:= \left\{\rho\in \overline{R(A_\kappa)}\,\big| \langle L_\kappa\psi, \rho\rangle=0\right\}.
\end{align*}
For any $\rho\in X_\kappa$ with $\int_{B_\kappa}\rho\,dx=\alpha$, we may write
\[
\rho = \alpha \psi + \beta\frac{d\rho_\kappa}{d\kappa} + \bar\rho, \ \ \langle L_\kappa\bar\rho,\psi\rangle=0, \ \ 
\bar\rho\in \overline{R(A_\kappa)}
\]
where 
\[
\beta:=-(1-\frac{2M_\kappa}{R_\kappa})\frac{\langle L_\kappa\rho,\psi\rangle - \alpha \langle L_\kappa\psi,\psi\rangle}{\frac{d}{d\kappa}\left(\frac{M_\kappa}{R_\kappa}\right)}.
\]
We conclude that 
\[
X_\kappa = S_1 \oplus S_2, \ \ \overline{R(A_\kappa)} = S_1 \oplus \mathbb R\frac{d\rho_\kappa}{d\kappa}.
\]
From a general functional analysis argument (Lemma 12.3 in~\cite{LinZengMemoirs}) it follows that 
\[
n^-(L_\kappa) = n^-(L_\kappa\big|_{S_0}) + n^-(L_\kappa\big|_{S_1})
\]
and from~\eqref{E:CASE1B} we have
\[
n^-(L_\kappa\big|_{\overline{R(A_\kappa)}}) = n^-(L_\kappa\big|_{S_1}).
\]
Since for any $\alpha,\beta\in\mathbb R$ we have 
\[
\langle L_\kappa(\alpha\psi+\beta\frac{d\rho_\kappa}{d\kappa}, \alpha\psi+\beta\frac{d\rho_\kappa}{d\kappa} )\rangle
=\alpha^2 \langle L_\kappa\psi,\psi\rangle - 2\alpha\beta \frac1{1-\frac{2M_\kappa}{R_\kappa}} \frac{d}{d\kappa}\left(\frac{M_\kappa}{R_\kappa}\right),
\]
it is clear that $n^-(L_\kappa\big|_{S_0})=1$. It thus follows that 
\[
n^-(L_\kappa\big|_{\overline{R(A_\kappa)}}) = n^-(L_\kappa) -1 = n^-(\Sigma_\kappa)-i_\kappa,
\]
where we have used Lemma~\ref{L:NEGMORSEINDEX} and Definition~\ref{D:INDEX}.
\end{enumerate}
\noindent
{\em Case 2.} Let $\frac{d}{d\kappa}\left(\frac{M_\kappa}{R_\kappa}\right)= 0$. 
In this case $\frac{d}{d\kappa}M_\kappa\neq 0$ by Lemma~\ref{L:NOCOEXISTENCE}. Therefore
$\int_{B_\kappa}\frac{d\rho_\kappa}{d\kappa}\,dx=\frac{d}{d\kappa}M_\kappa\neq0$ and thus 
$\frac{d\rho_\kappa}{d\kappa}\notin \overline{R(A_\kappa)}$. For this reason 
$X_\kappa =  \overline{R(A_\kappa)} \oplus \mathbb R\frac{d\rho_\kappa}{d\kappa}$
and thus $n^-(L_\kappa\big|_{\overline{R(A_\kappa)}})= n^-(L_\kappa)=n^-(\Sigma_\kappa)
= n^-(\Sigma_\kappa)-i_\kappa$, where we have used Lemma~\ref{L:NEGMORSEINDEX} and the Definition~\ref{D:INDEX}.

\noindent
{\em Proof of part (ii).}  We only sketch the proof as it is almost identical to the proof of part (ii) of Theorem 1.1 in~\cite{LinZeng2020}. We highlight one small difference. In order to show discreteness of the spectrum of $J^\kappa\mathcal L^\kappa$ we consider the space $Z_\kappa\subset Y_\kappa=L^2(B_\kappa)$ defined as the closure of with respect to the graph norm
\begin{align}
\|v\|_{Z_\kappa} &= \|v\|_{L^2(B_\kappa)} + \|A_\kappa v\|_{X_\kappa} \notag \\
& = \left(4\pi \int_0^{R_\kappa}|v|^2 r^2\,dr\right)^{\frac12} \notag \\
& \ \ \ \ + \left(4\pi \int_0^{R_\kappa}e^{2\mu_\kappa+\l_\kappa}\Psi_\kappa \left(\frac{1}{r^{2}} \frac{d}{dr}\left( r^{2}e^{\frac{\mu_\kappa-3\l_\kappa}{2}} (\rho_\kappa+p_\kappa)^{\frac{1}{2}}v\right)\right)^2 r^2\,dr\right)^{\frac12} \notag
\end{align}
The analogous versions of $Y_\kappa=L^2(B_\kappa)$ and $X_\kappa$-spaces in the Newtonian case as formulated in~\cite{LinZeng2020} are weighted by an additional power of $\rho_\kappa$, which affects their asymptotic rate of vanishing at the vacuum boundary. The discreteness follows if we can show that the embedding $Z_\kappa\hookrightarrow Y_\kappa$ is compact. This follows from Proposition 2.1 in~\cite{LinZeng2020}, which relies on the second order formulation~\eqref{E:SECONDORDER} and a general discreteness criterion (Theorem 4.2.9 in~\cite{EdEv}). The compactness of the embedding $Z_\kappa\hookrightarrow Y_\kappa$ is a consequence of Hardy's inequality and formula~\eqref{E:PSIKAPPAINVASYMP}, proof follows as in~\cite{LinZeng2020}.

\noindent
{\bf Proof of Theorem~\ref{T:TPP}.}
{\em Proof of part (i).}
We consider four cases.

\noindent
{\em Case 1.} 
Let $\bar\kappa>0$ be neither a critical point of $M_\kappa$ nor a critical point of $\frac{M_\kappa}{R_\kappa}$. This case is easy, as the signs of $\frac{d}{d\kappa}M_\kappa$
and $\frac{d}{d\kappa}\left(\frac{M_\kappa}{R_\kappa}\right)$ remain unchanged in a small neighbourhood of $\bar\kappa$. Therefore by Lemma~\ref{L:LOCCONST} and formula~\eqref{E:NEGATIVEMODES}, the number $n^u(\kappa)$ is constant in a neighbourhood of $\bar\kappa$.

\noindent
{\em Case 2.} 
Let $\bar\kappa>0$ be a critical point of $\frac{M_\kappa}{R_\kappa}$. By~\eqref{E:NEGATIVEMODES} and the jump formula~\eqref{E:JUMPFORMULA} we then have
\begin{align*}
n^u(\bar\kappa+) = n^-(\Sigma_{\bar\kappa+}) - i_{\bar\kappa,+} = n^-(\Sigma_{\bar\kappa-}) - i_{\bar\kappa,-}=n^u(\bar\kappa-).
\end{align*}
Therefore, $n^u(\kappa)$ remains constant in a neighbourhood of $\bar\kappa$.

\noindent
{\em Case 3.} 
Let $\bar\kappa$ be a local extremum of $M_\kappa$. By Lemma~\ref{L:NOCOEXISTENCE} we have $\frac{d}{d\kappa}\left(\frac{M_\kappa}{R_\kappa}\right)\Big|_{\kappa=\bar\kappa}\neq0$ and therefore by Lemma~\ref{L:LOCCONST} $n^-(\Sigma_\kappa)$ is constant in some neighbourhood of $\bar\kappa$. However, since $\frac{d}{d\kappa}M_\kappa$ has to change sign as $\kappa$ crosses $\bar\kappa$, we conclude from~\eqref{E:NEGATIVEMODES} 
\be\label{E:TPPAUX}
n^u(\bar\kappa+)- n^u(\bar\kappa-) = i_{\bar\kappa+}-i_{\bar\kappa-} = \pm1,
\ee
when $\frac{d}{d\kappa}M_\kappa\frac{d}{d\kappa}\left(\frac{M_\kappa}{R_\kappa}\right)$ changes sign from $\pm$ to $\mp$ as $\kappa$ increases and crosses $\bar\kappa$. Observe that $\frac{d}{d\kappa}R_\kappa\neq0$ in a neighbourhood of $\bar\kappa$ (otherwise $\frac{d}{d\kappa}M_\kappa\Big|_{\bar\kappa}=
\frac{d}{d\kappa}\left(\frac{M_\kappa}{R_\kappa}\right)\Big|_{\kappa=\bar\kappa}=0$, a contradiction to Lemma~\ref{L:NOCOEXISTENCE}). Moreover the sign of 
$\frac{d}{d\kappa}\left(\frac{M_\kappa}{R_\kappa}\right)$ is the same as the sign of $-\frac{d}{d\kappa}M_\kappa\frac{d}{d\kappa}R_\kappa$ in a small neighbourhood of $\bar\kappa$, which 
follows from the quotient rule and the bound
\[
\frac{\left\vert\frac{d}{d\kappa}M_\kappa\right\vert}{R_\kappa} \ll \frac{ \left\vert \frac{d}{d\kappa}R_\kappa\right\vert}{R_\kappa^2}, \ \ 
|\kappa-\bar\kappa|\ll1.
\]
In other words, formula~\eqref{E:TPPAUX} implies that $n^u(\bar\kappa+)- n^u(\bar\kappa-)=\pm 1$ when $\frac{d}{d\kappa}M_\kappa\frac{d}{d\kappa}R_\kappa$ changes sign from $\mp$ to $\pm$ as $\kappa$ passes through $\bar\kappa$, which is equivalent to the geometric statement that the mass-radius curve turns counter-clockwise, resp. clockwise, as $\kappa$ passes through $\bar\kappa$.

\noindent
{\em Case 4.} 
The remaining posibility is that $\bar\kappa$ is a critical point of $M_\kappa$, but not a local extremum. In this case 
$\frac{d}{d\kappa}\left(\frac{M_\kappa}{R_\kappa}\right)\big\vert_{\kappa=\bar\kappa}\neq0$ by Lemma~\ref{L:NOCOEXISTENCE} and therefore by Lemma~\ref{L:LOCCONST} the negative Morse index $n^-(\Sigma_\kappa)$ is constant in a small neighbourhood of $\bar\kappa$. Since $\bar\kappa$ is not an extremum of $\kappa\mapsto M_\kappa$ the sign of $\frac{d}{d\kappa}M_\kappa$ is also constant in a small neighbourhood of $\bar\kappa$. There are thus two possibilities. If $\frac{d}{d\kappa}M_\kappa\frac{d}{d\kappa}\left(\frac{M_\kappa}{R_\kappa}\right)>0$ in a small neighbourhood of $\bar\kappa$ then by~\eqref{E:KEYFORMULA}
\[
\langle L_\kappa \frac{d\rho_\kappa}{d\kappa}, \frac{d\rho_\kappa}{d\kappa} \rangle <0
\]
and for $\kappa\neq\bar\kappa$ $\int_{B_\kappa}\frac{d\rho_\kappa}{d\kappa}\,dx=\frac{d}{d\kappa}M_\kappa\neq0$, i.e. $\frac{d\rho_\kappa}{d\kappa}\notin 
\overline{R(A_\kappa)}$ by~\eqref{E:DYNACC}. This implies $n^u(\kappa)=n^-(L_\kappa\big|_{\overline{R(A_\kappa)}})=n^-(\Sigma_\kappa)-1$, where we have used Lemma~\ref{L:NEGMORSEINDEX}.
If on the other hand $\frac{d}{d\kappa}M_\kappa\frac{d}{d\kappa}\left(\frac{M_\kappa}{R_\kappa}\right)<0$ then 
$n^u(\kappa)=n^-(L_\kappa\big|_{\overline{R(A_\kappa)}})=n^-(L_\kappa)=n^-(\Sigma_\kappa)$, where we have used Lemma~\ref{L:NEGMORSEINDEX}. In both cases the formula~\eqref{E:NEGATIVEMODES} holds.

\noindent
{\em Proof of part (ii).}
It is well-known~\cite{NiUg, Ma2000, HeRoUg} that for the equations of state satisfying (P1)--(P4)
the mass-radius curve forms an infinite spiral which bends counter-clockwise as $\kappa\to\infty$. A simple consequence 
of part (i) is that 
\[
\lim_{\kappa\to\infty} n^u(\kappa)=\infty.
\]

\noindent
{\bf Sketch of the proof of Theorem~\ref{T:TRICHOTOMY}.}
In the non-degenerate case $\frac{d}{d\kappa}M_\kappa \frac{d}{d\kappa}\left(\frac{M_\kappa}{R_\kappa}\right) \neq0$ parts (i)-(iv) 
except for the bound~\eqref{bound-k-0-EE}, which in turn follows by the same argument as the corresponding Newtonian bound in~\cite{LinZeng2020}. The new contribution is the treatment of the degenerate case $\frac{d}{d\kappa}M_\kappa \frac{d}{d\kappa}\left(\frac{M_\kappa}{R_\kappa}\right) =0$, which can be treated verbatim as in~\cite{LinZeng2020}, where the role of the reduced operator is taken over by the relativistic reduced operator $\Sigma_\kappa$.


\section{Sufficient stability condition for the Einstein-Vlasov equilibria}\label{S:EV}


The unknowns in the EV-system are the Lorentzian manifold $(M,g)$ and the phase-space distribution function $f$ which is  supported on the mass-shell submanifold of the tangent bundle and solves the Vlasov equation. To find radially symmetric isotropic steady states, one prescribes a microscopic equation of state
\[
f = \Phi (1-\frac{E}{E^0})
\]
where $E$ is the local particle energy and $E^0$ some cut-off energy. 
Following \cite{RaRe,HaLiRe} we assume that $\Phi$ satisfies the assumption ($\Phi$1) (see (3.3) in~\cite{HaLiRe}) which we repeat here for completeness; we assume that  $\Phi\in L^\infty_{\text{loc}}([0,\infty))$ is a non-negative function, such that $\Phi(\eta)=0$ for all $\eta\le0$ and there exists a $-\frac12<k<\frac32$ and constants $c_1,c_2$ such that
\[
c_1\eta^k\le \Phi(\eta)\le c_2\eta^k, \ \ \text{for all $\eta\ge0$ sufficiently small}.
\]
For a fixed $\Phi$ satisfying these assumptions, 
by analogy to the Einstein-Euler system one obtains a 1-parameter family of steady states $\kappa\mapsto (f_\kappa,\mu_\kappa,\l_\kappa)$ of the asymptotically flat radial Einstein-Vlasov system (see Section 3 of~\cite{HaLiRe}) with finite ADM-mass
\[
\mathcal M(\rho)=4\pi \int_0^\infty \rho(r) \,r^2dr=4\pi \iint \sqrt{1+|v|^2}f(x,v)\,dv
\]
and compact support. Two important groups of examples that our result applies to are
\begin{align}
\Phi(x)&=x^k_+, \ \ -\frac12<k<\frac32, \ \ \text{(polytropes)}, \label{E:PHIPOLYTROPE}\\
\Phi(x)&=(e^{x}-1)_+, \ \ \text{(King's galaxy)}\label{E:PHIKING}
\end{align}

For such a family there is a canonical mapping $\Phi\mapsto P_\Phi$ (Section 3.2 of~\cite{HaLiRe}) which yields a macroscopic equation of state $\rho\mapsto P_\Phi(\rho)$ satisfying assumptions (P1)--(P4). For instance, using equations (3.5)-(3.6) in~\cite{HaLiRe} it is easy to see that in the small $0\le\rho\ll1$ region 
the Taylor expansion of $P_{\Phi}$ about $\rho=0$ takes the form
\[
P_{\Phi}(\rho) = c \rho^{\gamma_k} \left(1+ o_{\rho\to0^+}(1)\right), \ \ \gamma_k:=\frac{k+\frac52}{k+\frac32}.
\]
Our assumptions on the range of $k$ ensure that $\frac43<\gamma_k<2$ and therefore assumption (P2) is satisfied. It is easy to see that the remaining assumptions (P1),(P3)--(P4) also hold.
The resulting 1-parameter family of steady states $\kappa\to(\rho_\kappa,\mu_\kappa,\l_\kappa)$ of the Einstein-Euler system given by Proposition~\ref{P:EXISTENCE} has the {\em identical} mass-radius curve as the family $\kappa\mapsto (f_\kappa,\mu_\kappa,\l_\kappa)$. A simple corollary of Theorem~\ref{T:TPP} is then
\begin{theorem}[Sufficient stability condition for the Einstein-Vlasov equilibria] \label{T:TPPEV}
Let $\Phi$ satisfy the above assumptions. The 1-parameter family of steady states $\kappa\to(f_\kappa,\mu_\kappa,\l_\kappa)$ associated with the microscopic state function $\Phi$ is spectrally stable for all values of $\kappa\in(0,\kappa_{\text{max}})$, where $\kappa_{\text{max}}>0$ is the first maximum of the ADM-mass $\kappa\to \mathcal M_\kappa$.
\end{theorem}


\begin{proof}
Since the induced macroscopic equation of state $\rho\mapsto P_\Phi(\rho)$ satisfies assumptions (P1)--(P4), we apply Theorem~\ref{T:TPP} to conclude that $(\rho_\kappa,\mu_\kappa,\l_\kappa)$ are spectrally stable for $\kappa<\kappa_{\text{max}}$. By Theorem 5.26 in~\cite{HaLiRe} we conclude that $(f_\kappa,\mu_\kappa,\l_\kappa)$  are also spectrally stable for all $\kappa<\kappa_{\text{max}}$. Since for any $\kappa>0$ the ADM mass of $f_\kappa$ is the same as the ADM-mass of $\rho_\kappa$, the claim follows.
\end{proof}

{\bf Acknowledgements.}
The authors thank Gerhard Rein for helpful discussions.
M. H. acknowledges the support of the EPSRC Early Career Fellowship EP/S02218X/1.
Z. L. is supported partly by the NSF grants DMS-1715201 and DMS-2007457.

\end{document}